\documentclass{article}

\usepackage{arxiv}

\usepackage[utf8]{inputenc} 
\usepackage[T1]{fontenc}    
\usepackage{booktabs}       
\usepackage{amsfonts}       
\usepackage{nicefrac}       
\usepackage{microtype}      
\usepackage{lipsum}		
\usepackage{graphicx}
\usepackage{doi}

\date{Department of Computer Science\\ Purdue University} 					

\renewcommand{\headeright}{}
\renewcommand{\undertitle}{}
\renewcommand{\shorttitle}{}

\usepackage{lineno,hyperref}

\usepackage{microtype}

\usepackage{todonotes}

\usepackage{bbold}
\usepackage{graphics}

\usepackage{pst-node}
\usepackage{tikz-cd}
\usepackage[inline]{enumitem}
\usepackage{wrapfig,color}

\usepackage{amsmath}
\usepackage{amssymb}
\usepackage{amsthm}
\usepackage{xspace}
\usepackage[normalem]{ulem}
\usepackage{graphicx}
\usepackage{epsfig}
\usepackage{ifpdf}
\usepackage{url,hyperref}
\usepackage{latexsym}
\usepackage[mathscr]{euscript}
\usepackage{xspace}
\usepackage{color}
\usepackage{makeidx}
\usepackage{wrapfig,color}
\usepackage{stackrel}
\usepackage[linesnumbered,algoruled,boxed,lined]{algorithm2e}
\usepackage{bbm}

\usepackage[all]{xy}

\usepackage{mathtools}

\DeclarePairedDelimiter\set{\{}{\}}

\theoremstyle{plain}
\newtheorem{theorem}{Theorem}[section]
\newtheorem{proposition}[theorem]{\bf Proposition}

\theoremstyle{definition}
\newtheorem{definition}{Definition}[section]
\newtheorem{fact}{Fact}
\newtheorem{lemma}[theorem]{Lemma}
\newtheorem{remark}{Remark}[section]

\newtheorem*{theorem*}{Theorem}
\newtheorem*{proposition*}{Proposition}
\newtheorem*{lemma*}{Lemma}
\newtheorem*{note*}{Note}
\newtheorem*{remark*}{Remark}
\newtheorem*{notation*}{Notation}

 \newcommand{\propositionofref}{}
\newtheorem*{zpropositionof}{Proposition \propositionofref}
\newenvironment{propositionof}[1]
 {\renewcommand{\propositionofref}{#1}\zpropositionof}
 {\zpropositionof}

 \newcommand{\theoremofref}{}
\newtheorem*{ztheoremof}{Theorem \theoremofref}
\newenvironment{theoremof}[1]
 {\renewcommand{\theoremofref}{#1}\ztheoremof}
 {\ztheoremof}

\newtheorem{corollary}[theorem]{Corollary}
\newtheorem{claim}{Claim}[theorem]

\newcommand{\push}{\mathrm{push}}

\newcommand{\Real}{\mathbb{R}}
\newcommand{\Rn}{\Real^n}

\newcommand{\VecSp}{\mathbf{Vec}}
\newcommand{\vecsp}{\mathbf{vec}}
\newcommand*{\rom}[1]{\expandafter\@slowromancap\romannumeral #1@}
\newcommand{\Img}{\mathrm{im\,}}
\newcommand{\coim}{\mathrm{coim\,}}

\newcommand{\dl}{\mathrm{dl}}
\newcommand{\di}{d_\mathrm{{I}}}
\newcommand{\tdi}{\tilde{d}_{\mathrm{I}}}
\newcommand{\dH}{d_\mathrm{{H}}}
\newcommand{\db}{d_\mathrm{{B}}}
\newcommand{\tdb}{\tilde{d}_{\mathrm{B}}}
\newcommand{\dmatch}{d_\mathrm{{match}}}

\newcommand{\dbmm}{\tilde{d}_{\mathrm{B}}}
\newcommand{\Poset}{\mathbb{P}}
\newcommand{\nmod}{\Rn\mbox{-}\mathbf{mod}}
\newcommand{\pindexmod}{\mathbb{P}\mbox{-}\mathbf{mod}}
\newcommand{\indexmod}{\mbox{-}\mathbf{mod}}
\newcommand{\RB}{\bar{\mathbb{R}}}
\newcommand{\field}[1] {{\mathbb{#1}}}
\newcommand{\row}{{\mathrm{row}}}
\newcommand{\col}{{\mathrm{col}}}

\newcommand{\M}{M}
\newcommand{\N}{N}
\newcommand{\LL}{\mathbf{L}}
\newcommand{\p}{\mathbf{p}}
\newcommand{\q}{\mathbf{q}}

\newcommand{\rr}{\mathbf{r}}
\newcommand{\s}{\mathbf{s}}
\newcommand{\SSet}{\mathbf{S}}
\newcommand{\x}{\mathbf{x}}
\newcommand{\y}{\mathbf{y}}
\newcommand{\z}{\mathbf{z}}
\newcommand{\uu}{\mathbf{u}}
\newcommand{\vv}{\mathbf{v}}
\newcommand{\ww}{\mathbf{w}}
\newcommand{\aaa}{\mathbf{a}}

\newcommand{\CC}{\mathcal{C}}

\newcommand{\gen}{\mathrm{Gen}}
\newcommand{\rel}{\mathrm{Rel}}
\newcommand{\gr}{\mathrm{gr}}

\newcommand{\PAR}{\mathrm{Cov}}

\newcommand{\spanning}{\mathrm{Span}}
\newcommand{\diam}{\mathrm{Diam}}

\definecolor{darkred}{rgb}{1, 0.1, 0.3}
\definecolor{darkblue}{rgb}{0.1, 0.1, 1}

\DeclareMathOperator*{\argmin}{arg\,min}
\DeclareMathOperator{\sign}{sgn}

\title{Rectangular Approximation and Stability of $2$-parameter Persistence Modules}

\author{
    Tamal K. Dey \thanks{This research is supported by NSF grants CCF 1740761 and 2049010}\\
\texttt{tamaldey@purdue.edu} 
\And
Cheng Xin\\
\texttt{xinc@purdue.edu} 
}

\begin{document}
\maketitle










\begin{abstract}
    One of the main reasons for topological persistence being useful in data analysis is that it is backed up by a stability (isometry) property: persistence diagrams of $1$-parameter persistence modules are stable in the sense that the bottleneck distance between two diagrams equals the interleaving distance between their generating modules. However, in multi-parameter setting this property breaks down in general. A simple special case of persistence modules called rectangle decomposable modules is known to admit a weaker stability property.
    Using this fact, we derive a stability-like property for $2$-parameter persistence modules.
    For this, first we consider interval decomposable modules and their optimal approximations with rectangle decomposable modules with respect to the bottleneck distance. We provide a polynomial time algorithm to exactly compute this optimal approximation which, together with the polynomial-time computable bottleneck distance among interval decomposable modules, provides a lower bound on the interleaving distance.
    Next, we leverage this result to derive a polynomial-time computable distance for general multi-parameter persistence modules which enjoys similar stability-like property. This distance can be viewed as a generalization of the matching distance defined in the literature.
\end{abstract}






\section{Introduction} \label{sec: intro}
Stability of persistence diagrams~\cite{Chazal:2009,cohen2007stability} for $1$-parameter persistence modules provides an assurance in topological data analysis that data corrupted possibly with measurement errors and noise still can give insights into the topological features of the sampled space. This stability result hinges on a distance called the interleaving distance $\di$ between persistence modules~\cite{Chazal:2009} and another distance called the bottleneck distance $\db$ between persistence diagrams~\cite{cohen2007stability} that describe the isomorphism type of a (sufficiently tame) persistence module combinatorically. First, Chazal et al. showed that $\db\leq \di$~\cite{Chazal:2009} and later Lesnick showed that indeed $\db=\di$~\cite{lesnick2015theory}. Since $\db$ is computable by polynomial time algorithms, we can compute $\di$ in polynomial time as well. Unfortunately, $\di$ cannot be computed in polynomial time for $2$-parameter persistence modules unless P=NP as shown by Bjerkevik et al.~\cite{bjerkevik18computational_complexity,Bjerkevik2020}. Even approximation of $\di$ with a constant factor better than 3 is shown to be NP-hard. Therefore, the search is on to find out whether $\di$ can be approximated with additive factors as opposed to multiplicative ones, or whether there are interesting special cases where $\di$ can be approximated with constant multiplicative factors.

Guaranteed by Krull-Schmidt theorem~\cite{Atiyah1956}, a persistence module enjoys an essentially unique decomposition. One can define a bottleneck distance $\db$ (see Definition~\ref{def:bottleneck_distance}) based on the interleaving distance $\di$ and a partial matching between these decompositions.
It follows from the definition that $\di\leq \db$~\cite{botnan2016algebraic}. So, if one can establish that $\db\leq c\cdot \di$ for some constant $c\geq 1$ and $\db$ is polynomial time computable, we can have a useful stability property in practice. Toward this goal Bjerkevik~\cite{bjerkevik2021stability} showed that $\db\leq (2n-1)\di$ for a special class of $n$-parameter modules whose indecomposables are rectangles and hence called rectangle decomposable modules. However, no such lower bound exists for general persistence modules as has been observed in~\cite{botnan2016algebraic}. So, in presence of this negative result, one can only hope for computing bottleneck distance with some subtractive factors
which bound $\di$ from below. This is what we do in this paper.

In an earlier paper~\cite{DX18}, we have presented a polynomial time algorithm for computing the bottleneck distance for a class of modules called $2$-parameter interval decomposable persistence modules~\cite{bjerkevik2021stability}. We give a polynomial time algorithm to approximate these modules with rectangle decomposable modules optimally. This optimality is measured with respect to some distance that we define later.  Given two $2$-parameter interval decomposable modules $M$ and $N$ and their optimal approximations with rectangle decomposable modules
$\boxed{M}^*$ and $\boxed{N}^*$ respectively, we show that (Theorem~\ref{thm:db_bound1})
\begin{equation*}
    \frac{1}{3}\db(M,N)-\frac{4}{3}\left(\db(M,\boxed{M}^*)+\db(N,\boxed{N}^*)\right) \leq \di(M,N).
\end{equation*}
Since all quantities on the left can be computed in polynomial time, we can compute
a non-trivial lower bound on $\di$ in polynomial time.

 We extend the result to $2$-parameter general persistence modules though not achieving as good an approximation. For this, we partition the support of the given modules so that each module restricted to every component of the partition becomes interval decomposable. We construct a family of distances that generalize the so called matching distance. It is bounded from above by the bottleneck distance (Proposition~\ref{prop:tdb_stability}). The stability of this distance with respect to $\di$ follows from our previous result.

\section{Background}

\begin{definition}[Persistence module: categorical definition]
    Let $\Poset$ be a poset category. A $\Poset$-indexed persistence module is a functor $\M:\Poset\rightarrow \VecSp$ where $\VecSp$ is the category of vector spaces over some field $\mathbb{k}$. If $\M$ takes values in $\vecsp$, the category of finite dimensional vector spaces, we say $\M$ is pointwise finite dimensional (p.f.d).
  The $\Poset$-indexed persistence modules themselves form a category of functors, denoted as $\pindexmod$, where the natural transformations between functors constitute the morphisms. 
\end{definition}

For a subposet $\mathbb{B} \subseteq \Poset$, there is a canonical restriction functor from $\pindexmod$ to its subcategory $\mathbb{B}\indexmod$ by restriction. For $M\in\pindexmod$, denote its restriction on $\mathbb{B}$ as $M|_\mathbb{B}$.

Here we consider the poset category on $\Rn$ or some convex subset $C\subseteq\Rn$, with the standard partial order and requiring all modules to be p.f.d. We also fix $\field{k}=\mathbb{F}_2$. We call $\Rn$-indexed persistence modules as $n$-parameter persistence modules. The category of $n$-parameter modules is denoted as $\nmod$. For an $n$-parameter module $M\in \nmod $, we also use notations $M_\x:=M(\x)$ and $M({\x\rightarrow \y}):=M(\x\leq \y)$.

\begin{definition}[Shift] For any $\delta \in \mathbb{R}$, we denote  $\vec{\delta}=(\delta, \cdots, \delta)$. 
We define a shift functor $(\cdot)_{\rightarrow \delta}:\nmod \rightarrow \nmod$ where $M_{\rightarrow \delta}:=(\cdot)_{\rightarrow \delta} (M)$ is given by $M_{\rightarrow \delta}(\x)=M(\x+\vec{\delta})$ 
and $M_{\rightarrow\delta}(\x\leq \y)=M(\x+\vec{\delta} \leq \y+\vec{\delta})$. In other words, $M_{\rightarrow \delta}$ is the module $M$ shifted diagonally by $\vec\delta$.
\end{definition}

The following definition of interleaving taken from ~\cite{oudot2015persistence} adapts the original definition designed for 1-parameter modules in~\cite{chazal2016structure} to $n$-parameter modules.
\begin{definition}[Interleaving]
For two persistence modules $M$ and $N$, and $\delta \geq 0$, a \emph{$\delta$-interleaving} between $M$ and $N$ are two families of linear maps $\{\phi(\mathbf{x}): M(\mathbf{x}) \rightarrow N({\mathbf{x}+\vec{\delta}})\}_{\mathbf{x}\in \mathbb{R}^n}$ and $\{\psi(\mathbf{x}): N(\mathbf{x}) \rightarrow M({\mathbf{x}+\vec{\delta}})\}_{\mathbf{x}\in \mathbb{R}^n}$ satisfying the following two conditions:

\begin{itemize}
    \item $\forall \mathbf{x} \in \mathbb{R}^n, M({\mathbf{x} \rightarrow \mathbf{x}+2\vec{\delta}}) = \psi({\mathbf{x}+\vec\delta}) \circ \phi({\mathbf{x}})$ and
    $N({\mathbf{x} \rightarrow \mathbf{x}+2\vec{\delta}}) = \phi({\mathbf{x}+\vec\delta}) \circ \psi({\mathbf{x}}),$
    \item $\forall \mathbf{x}\leq \y \in \mathbb{R}^n, 
     \phi({\y}) \circ M(\mathbf{x} \rightarrow \y) =  N(\mathbf{x}+\vec{\delta} \rightarrow \y+\vec{\delta}) \circ \phi({\mathbf{x}}) $ and   
    $\psi({\y}) \circ N({\mathbf{x} \rightarrow \y}) =  M({\mathbf{x}+\vec{\delta} \rightarrow \y+\vec{\delta}}) \circ \psi({\mathbf{x}}).$
\end{itemize}

If such a $\delta$-interleaving exists, we say $M$ and $N$ are $\delta$-interleaved. 
\end{definition}

\begin{definition}[Interleaving distance]
    The interleaving distance between modules $M$ and $N$ is defined as  $\di(M, N)=\inf_{\delta}\{M \mbox{ and } N \mbox{ are } \delta\textit{-interleaved}\}$. 
    We say $M$ and $N$ are $\infty$-interleaved if they are not $\delta$-interleaved for any positive $\delta\in \Real$, and assign $\di(M, N)=\infty$.
\end{definition}
Following~\cite{botnan2016algebraic}, we call a module $M$ \emph{$\delta$-trivializable} if $\delta \geq \di(M,0)$.





\begin{definition}[Matching]
    A matching  $\mu: A \nrightarrow B$ between two multisets $A$ and $B$ is a partial bijection, that is, $\mu: A' \rightarrow B'$ for some $A' \subseteq A$ and $B' \subseteq B$. We say $\Img\mu = B', \coim\mu = A'$.
\end{definition}


\begin{definition}[Indecomposable]
    We say a module $M$ is \emph{indecomposable} if $M\simeq M_1\oplus M_2 \implies M_1=0$ or $M_2=0$.
\end{definition}

By the Krull-Schmidt theorem~\cite{Atiyah1956},
there exists an essentially unique (up to permutation and isomorphism) decomposition $M\simeq \bigoplus M_i$ with every $M_i$ being indecomposable. 

\begin{definition}[Bottleneck distance]\label{def:bottleneck_distance}
Let $M \cong \bigoplus_{i=1}^{m} M_{i}$ and $N\cong \bigoplus_{j=1}^{n} N_{j}$ be two persistence modules, where $M_{i}$ and  $N_{j}$ are indecomposable submodules of $M$ and $N$ respectively. Let $I=\{1,\dotsb,m\}$ and $J=\{1,\dotsb,n\}$. We say $M$ and $N$ are $\delta$-matched for $\delta \geq 0$ if there exists a matching
$\mu:I \nrightarrow J$ so that,
(i) $i\in I\setminus \coim\mu \implies M_{i} \mbox{ is } \delta$-trivializable,
(ii) $j\in   J\setminus \Img\mu \implies N_{j} \mbox{ is } \delta$-trivializable, and
(iii) $i\in            \coim\mu \implies \di( M_{i},N_{\mu(i)})\leq\delta$.

The bottleneck distance is defined as
\[
\db(M, N) = \inf\{\delta\mid M \mbox{ and } N \mbox{ are }\delta \mbox{-matched}\}.
\]
\end{definition}
Note that this definition of bottleneck distance works in general for any persistence modules as long as the decomposition and interleaving distance are well defined. Also, from the definition it is easy to observe the following fact:
\begin{fact}
$\di \leq \db$.
\end{fact}

\subsection{Interval decomposable modules}
\begin{sloppypar}
 Persistence modules whose indecomposables are interval modules (Definition \ref{interval-def}) are called {\em interval decomposable modules}, see for example \cite{botnan2016algebraic}. To account for the boundaries of free modules, we enrich the poset $\Rn$ by adding points at $\pm\infty$ and consider the poset $\bar{\Real}^n=\bar{\Real}\times\ldots\times\bar{\Real}$ where $\bar{\Real}=\Real\cup\set{\pm\infty}$ with the additional rules $a\pm\infty=\pm\infty$, $\infty-\infty=0$.
\end{sloppypar}

\begin{definition}
An interval is a subset $\emptyset \neq I \subset \bar{\Real}^n$ that satisfies the following:
\begin{enumerate}
\item If $\p,\q \in I$ and $\p \leq \rr \leq \q$, then $\rr \in I$;
\item If $\p,\q \in I$, then there exists a finite sequence ($\p_1, \p_2, ... , \p_{m}) \in I$ such that $\p\leq \p_1 \geq \p_2 \leq \p_3 \geq ... \geq \p_{m} \leq q$.
\end{enumerate}


\end{definition}

Let $\bar{I}$ denote the closure of an interval $I$ in the standard topology of $\RB^n$. 
The lower and upper boundaries of $I$ are defined as
\begin{eqnarray*}
L(I)&=&\set{\x=(x_1,\cdots, x_n)\in \bar{I}\mid \forall \y=(y_1,\cdots, y_n) \mbox{ with } y_i< x_i \; \forall i \implies \y\notin I}\\
U(I)&=&\set{\x=(x_1,\cdots, x_n)\in \bar{I}\mid \forall \y=(y_1,\cdots, y_n) \mbox{ with } y_i> x_i \; \forall i \implies \y\notin I}.
\end{eqnarray*}

Following Section 6.1 in \cite{miller2017data}, we define
the boundary $B(I)$ of $I$ as $B(I)=L(I)\cup U(I)$. 
The vertex set $V(I)$ consists of all corner points in $B(I)$.
For example, 
$\RB^n$ is an interval with boundary $B(\RB^n)$ that consists of all the points with at least one coordinate $\infty$. 
The vertex set $V(\RB^n)$ consists of $2^n$ corner points of the infinitely large 
cube $\RB^n$ with coordinates $(\pm\infty,\cdots, \pm\infty)$.

\begin{definition}[Interval module]
\label{interval-def}
An $n$-parameter \emph{ interval persistence module}, or \emph{interval module} in short, is a persistence module $M$ that satisfies the following condition: for some interval $I_M\subseteq \RB^n$, called the interval of $M$, 
\begin{equation*}
M(\x) =
\begin{cases}

\mathbb{k} & \mbox{if $\x \in I_M$}\\
0 & otherwise
\end{cases}
\qquad 
M({\x \rightarrow \y}) =
\begin{cases}
 \mathbb{1} & \mbox{if $\x,\y \in I_M $}\\
0 & otherwise.
\end{cases}
\end{equation*}

\end{definition}


It is known that an interval module is indecomposable \cite{lesnick2015theory}.
%

\begin{definition}[Interval decomposable module]
An $n$-parameter {\em interval decomposable module} is a persistence module that can be decomposed as a direct sum of interval modules.
\end{definition}

\begin{definition}[Rectangle]
For some $\uu \leq \vv\in \RB^n$, we say the set $R=\{\ww\mid\uu\leq\ww\leq\vv\}\subseteq \RB^n$ is a rectangle in $\RB^n$, denoted as $R=[\uu, \vv]$.
\end{definition}

\begin{definition}[Rectangle decomposable module]
A \emph{rectangle module} is an interval module with underlying interval being a rectangle.
An $n$-parameter {\em rectangle decomposable module} is an interval decomposable module with all indecomposable components being rectangle modules.
\end{definition}


Bjerkevik~\cite{bjerkevik2021stability} proves the following algebraic stability property.

\begin{theorem}[Bjerkevik~\cite{bjerkevik2021stability}]\label{thm:rectangle_stability}
For two p.f.d. rectangle decomposable $\RB^n$-modules $M, N$, $\db(M, N)\leq (2n-1)\di(M, N)$. 
\end{theorem}
\begin{remark}
Combined with $\di\leq\db$ for general persistence modules, we have $\di(M, N)\leq \db(M, N)\leq (2n-1)\di(M, N)$ for rectangle decomposable modules. Specifically, 
when $n=1$, it becomes the isometry theorem, and 
when $n=2$, $\di(M, N)\leq \db(M, N)\leq 3\di(M, N)$ for 2-parameter rectangle decomposable modules. 
\end{remark}



\begin{definition} [Intersection module]
    For two interval modules $M$ and $N$ with intervals $I_M$ and $I_N$ respectively
    let $I_Q = I_M\cap I_N$, which is a disjoint union of intervals, $\coprod I_{Q_i}$.  
    The intersection module $Q$ of $M$ and $N$ is $Q=\bigoplus Q_i$, where $Q_i$ is the interval module with interval $I_{Q_i}$. That is,
    \begin{equation*}
    Q(x)= 
    \begin{cases}
   \field{k}  & \mbox{if $x\in I_M\cap I_N$}\\
    0   & otherwise
    \end{cases} \quad \mbox{and for $x\leq y,\,$ }
    {Q}({x\rightarrow y})=
    \begin{cases}
    \mathbb{1}  & \mbox{if $x,y \in I_M \cap I_N$}\\
    0   & otherwise.
    \end{cases}
    \end{equation*}
\end{definition}


\begin{definition} [Diagonal projection and distance]
  Let $I$ be an interval and $\x\in \RB^n$. Let $\Delta_\x=\{\x+ \vec{\alpha}\mid  \alpha \in \Real \}$ denote the \emph{diagonal line} passing through $\x$.
When $\Delta_\x\cap I \neq \emptyset$, we define $\pi_I(\x)\triangleq \argmin_{\y\in \Delta_\x\cap I}\set{ \|\x-\y\|_{\infty}}$, called the {\em projection point} of $\x$ on $I$. 
  We define (see Figure~\ref{fig:dl}) 
    \begin{equation*}
        \dl(\x,I)= 
        \begin{cases}
            \|\x-\pi_I(\x)\|_{\infty} \mbox{ if } \Delta_\x\cap I \neq \emptyset
          \\
          +\infty \mbox{ otherwise.}
        \end{cases}
    \end{equation*}
and its signed version
    \begin{equation*}
    \dl_{\pm}(\x,I)= 
    \begin{cases}
        \sign(\pi_I(\x)-\x)\cdot\dl(\x, I) \mbox{ if } \Delta_\x\cap I \neq \emptyset\\ 
        +\infty \mbox{ otherwise.}
    \end{cases}
    \end{equation*}
    where 
    $\sign(\x-\pi_I(\x))=\sign((\x-\pi_I(\x))\cdot(1,\dots, 1))$ 
    is used to indicate the relative positions of $\x$ and $\pi_I(\x)$. 
    Specifically, for $\mathbf{x}=(x_1, x_2)$, $\dl_\pm(\mathbf{x}, I)=a-x_1$ when $I$ is a vertical line $\{x=a\}$, and $\dl_\pm(\mathbf{x}, I)=b-x_2$ when $I$ is a horizontal line $\{y=b\}$. 
\end{definition}


\begin{note*}
Note that $\x\in I \iff \pi_I(\x)=\x \iff \dl(\x, I)=0$.
$\forall \alpha\in \Real, \quad  \pm \infty + \alpha = \pm \infty$. Therefore, for $\x\in V(\RB^n)$, the diagonal line  $ \Delta_{\x}$ collapses to a single point. In that case, $\dl(\x, I)\neq+\infty$ if and only if $\x\in I$, which means $\pi_I(\x)=\x$.
\end{note*}

\section{Rectangular approximation}\label{sec:rectangle_approx}

From now on, we focus on $2$-parameter p.f.d. persistence modules though many of our results can be generalized to multi-parameter persistence modules. For brevity, sometimes we use the notations $a\vee b\triangleq \max\{a, b\}$ and $a\wedge b\triangleq \min\{a,b\}$. Note that these two operations together are not associative.

\subsection{Computing an approximation}


For an interval module $M$ with the underlying interval being a polygon described by its set of vertices $V(I_M)$ in $\Real^2$,
we want to compute a rectangle module which is optimally close to $M$ under some distance. 
We assume all intervals include boundaries. If not, we replace $M$ with $\overline{M}$ by adding the boundary. This operation does not change the interleaving distance since $\di(M, \overline{M})=0$ as shown in~\cite{DX18}. One can also use decorated number~\cite{chazal2016structure} to avoid the technical issues related to boundaries.
We first give a construction that may not be optimal, but can be computed fast in time linear in input size $|V(I_M)|$. Next, we modify the construction further to compute an optimal approximation in polynomial time though with increased time complexity. We first define the optimality.
\begin{definition}
Given a persistence module $M$, a distance $d(\cdot)$ on persistence modules and a class of persistence modules $\mathcal{N}$, we define $d(M, \mathcal{N})\triangleq\inf\{d(M, N)\mid N\in \mathcal{N}\}$. 
We say a persistence module $N$ from a class of persistence modules $\mathcal{N}$ approximates a given persistence module $M$ optimally with respect to a distance function $d(\cdot)$ if $d(M, N)=d(M, \mathcal{N})$
\end{definition}

For an interval module $M$, we say a rectangular subset $J\subseteq I_M$ is a maximal rectangle of $I_M$ if it is not properly contained in any other rectangular subset in $I_M$. 
For each $I_M$, there is a unique maximal rectangle incident to the top-left corner of $I_M$, denoted as $I_M^\top$, and a unique maximal rectangle incident to the bottom-right corner of $I_M$, denoted as $I_M^\bot$. See Figure~\ref{fig:top_bot_max_rectangle} for an illustration.

\begin{figure}[htbp]
\centerline{\includegraphics[width=0.5\textwidth]{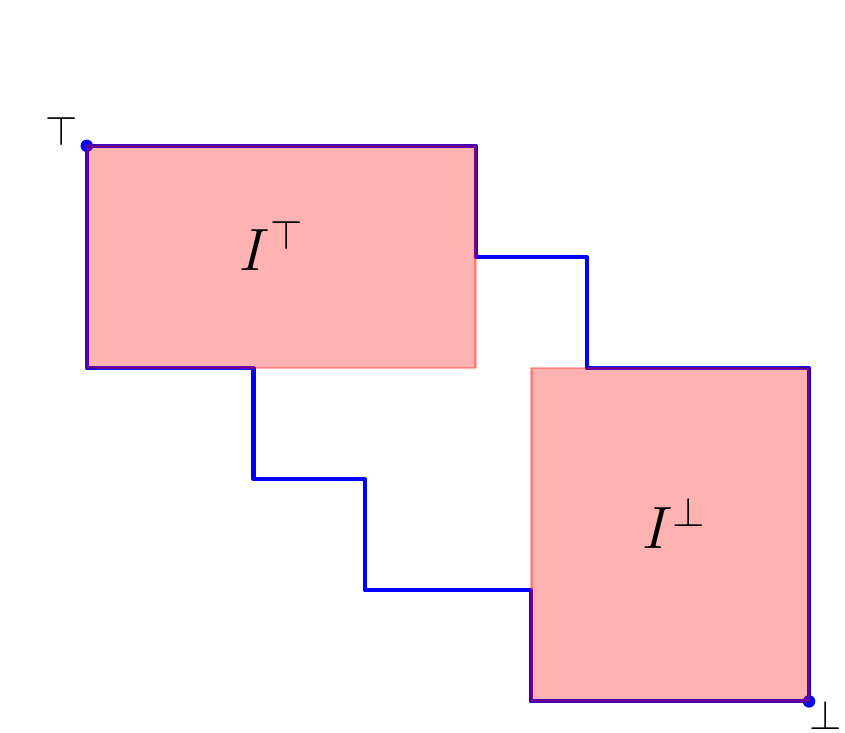}}
\caption{Red rectangles illustrate $I_M^\top$ and $I_M^\bot$ of interval $I$.}
\label{fig:top_bot_max_rectangle}
\end{figure}

\begin{definition}[Construction 1]\label{def:construction1}
For an interval $I_M$, we call the smallest rectangle containing it as the {\emph{bounding rectangle}} of $I_M$. 
Let $[\rr', \s']$ be the bounding rectangle of $I_M$. Let $\rr''=\pi_{L(I_M)}(\rr')$ and $\s''=\pi_{U(I_M)}(\s')$. Let $\mathbf{r}=\frac{1}{2} (\mathbf{r}'+\mathbf{r}'')$ and $\mathbf{s}=\frac{1}{2} (\mathbf{s}'+\s'')$. Let 
\begin{equation*}
\epsilon_M=\max\{ \|\mathbf{r}-\mathbf{r'}\|_{\infty}, \|\mathbf{s}-\mathbf{s'}\|_{\infty}\}.
\end{equation*}
Define $\boxed{M}$ to be the rectangle module on $[\mathbf{r}, \mathbf{s}]$ if $M$ is not $\epsilon_M$-trivializable. Define $\boxed{M}=0$, the trivial module, otherwise. See Figure~\ref{fig:construction1} for an illustration.
\end{definition}

\begin{figure}[ht]
\centerline{\includegraphics[width=0.5\textwidth]{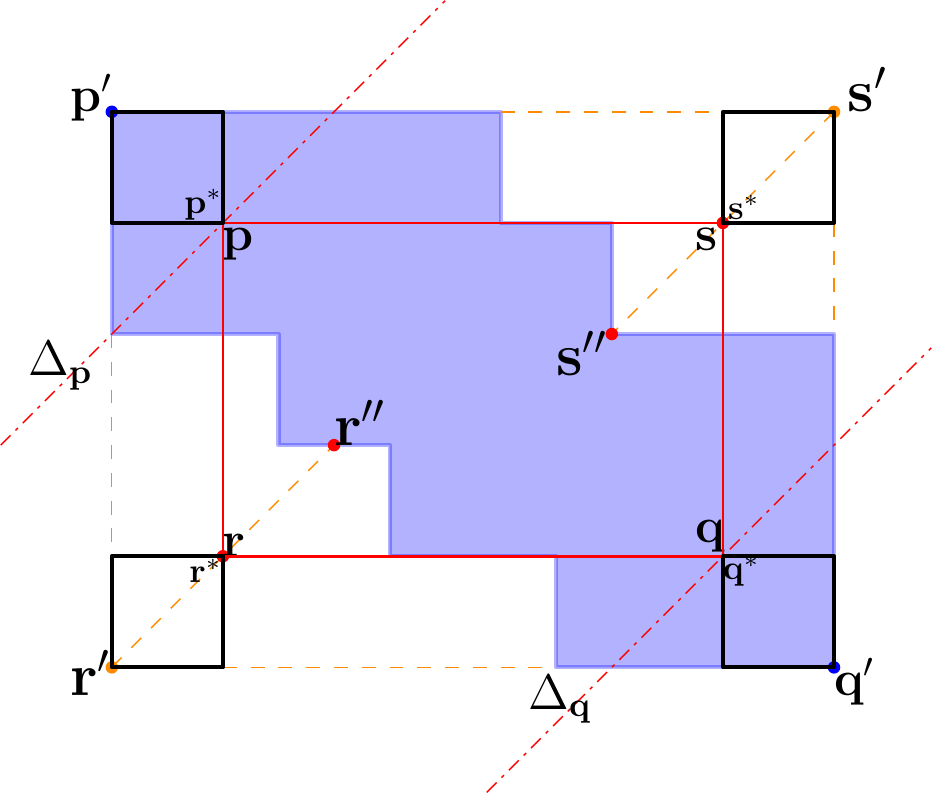}}
\caption{$M=$Blue region. $\boxed{M}=[\rr, \s]=$red rectangle. The rectangle given by $[\rr', \s']$ is the bounding rectangle of $I_M$. Under the assumption that both $M^\top$ and $M^\bot$ are not $\epsilon_M$-trivializable, the optimal rectangle $R^*$ should have four corners contained in the four black boxes corresponding to the four corners of the bounding rectangle, in order to make sure $M|_{\bar{\Delta}_{R^*}}$ is $\epsilon_M$-trivializable.}\label{fig:construction1}
\end{figure}

We have the following properties for $\boxed{M}$ (proof in \ref{sec:missing_proofs}).
\begin{theorem}\label{prop:properties_box}
The module $\boxed{M}$ obtained by {\sf Construction 1} satisfies the following properties:
\begin{enumerate}
    \item $\di(M, \boxed{M})\leq \epsilon_M$. 
    \item $\di(M, \boxed{M})=\epsilon_M$ if $\di(M, 0)\geq \epsilon_M$.
    \item $\boxed{M}$ is an optimal rectangle module approximating $M$ with respect to $\di$ if both $I_M^{\top}$ and $I_M^{\bot}$ are large enough to contain a $2\epsilon_M\times2\epsilon_M$ square. 
\end{enumerate}
\label{construct1-thm}
\end{theorem}

We extend our rectangular approximation to interval decomposable modules. 
For an interval decomposable module $M=\oplus M_i$, we extend the definition of $\boxed{M}$ to be $\boxed{M}\triangleq \bigoplus \boxed{M_i}$, and let $\epsilon_M\triangleq \max\{\epsilon_{M_i}\}$. We have the following property, the first one of which follows directly from the definitions.

\begin{proposition}\label{prop:db_bound}
For interval decomposable module $M$, $\di(M, \boxed{M})\leq \db(M, \boxed{M})\leq  \epsilon_M.$
\end{proposition}

\begin{proposition}\label{prop:db_bound0}
For two interval decomposable modules $M$ and $N$, 
\begin{equation*}
\di(M,N)\leq \db(M,N) \leq \db(\boxed{M}, \boxed{N}) + \epsilon_M + \epsilon_N \leq 3\di(M, N) + 4(\epsilon_M + \epsilon_N).
\end{equation*}
\end{proposition}

\begin{proof}
The first inequality is true for any persistence modules. The second inequality comes from triangle inequality. By the stability property of rectangle decomposable modules in Theorem~\ref{thm:rectangle_stability}, 
$\di(\boxed{M}, \boxed{N})\leq \db(\boxed{M}, \boxed{N})\leq 3\di(\boxed{M}, \boxed{N})$. Using the triangle inequality and Proposition~\ref{prop:db_bound}, we get $\di(\boxed{M}, \boxed{N})\leq \di(M,N) +\epsilon_M + \epsilon_N$.  The last inequality immediately follows.
\end{proof}

By Proposition~\ref{prop:db_bound}, $\db(M, \boxed{M})$ and $\db(N, \boxed{N})$ are bounded from above by $\epsilon_M$ and $\epsilon_N$ respectively. Then the question is, what is a closest rectangle decomposable module to $M$ with respect to the bottleneck distance. 
In line with our previous definition of optimality, 
we denote $\boxed{M_i}^*$ to be an optimal rectangle module of the interval module $M_i$ under interleaving distance, and set $\boxed{M}^*\triangleq \oplus \boxed{M_i}^*$ to be a rectangle decomposable module approximating $M\simeq \oplus M_i$. 
First we have the following observation which follows from the definifition of
the bottleneck distance:
\begin{proposition}
Given an interval decomposable module $M\simeq \oplus M_i$, a rectangle decomposable module $\boxed{M}^*\triangleq \oplus \boxed{M_i}^*$ approximates $M$ optimally with respect to the bottleneck distance. 
\end{proposition}

Therefore, to find an optimal rectangle decomposable module for $M\simeq \oplus M_i$ under the bottleneck distance, the question becomes how to find an optimal rectangle approximation for each interval module $M_i$. 

\subsection{An optimal approximation}

In {\sf Construction 1} (Definition~\ref{def:construction1}), we described an approximating rectangle module $\boxed{M}$ for an interval module $M$. By Theorem~\ref{prop:properties_box}, we 
can observe from Property 3 that $\boxed{M}$ fails to be $\boxed{M}^*$ in general case when the top or bottom maximal rectangles of $M$, that is $M^{\top}$ or $M^{\bot}$, is $\epsilon$-trivializable with $\epsilon$ being quite small. In that case, it is more reasonable to ignore these $\epsilon$-trivializable parts to construct a rectangular approximation on the remaining part possibly improving the approximation. Therefore, to compute $\boxed{M}^*$, we need to determine what are the top and bottom parts we want to ignore. 
In this section, we propose an algorithm to compute $\boxed{M}^*$ in general case.


First we introduce the Hausdorff distance in infinity norm which is useful in our context. 
For any set $S\subseteq\Real^2$ and $\delta \geq 0$, we denote 
\begin{equation*}
S^{(+\delta)}=\bigcup_{\x\in S}\{\x'\in\Real^2: \|\x-\x'\|_\infty\leq \delta\}.
\end{equation*}
The Hausdorff distance with the infinity norm between two nonempty sets $S,T\subseteq \Real^n$ is: 
\begin{equation*}
\dH(S,T)=\inf_{\delta\geq 0}\{S\subseteq T^{(+\delta)} \And T\subseteq S^{(+\delta)}  \}.  
\end{equation*}

For two nontrivial interval modules $\M, \N$, we denote $\dH(\M,\N)\triangleq \dH(I_\M, I_\N)$. 
If $\N=0$, we set $\dH(\M, \N)\triangleq \di(\M, 0)$. Given any interval module $\M$ and a rectangle module $R$ with $I_R$ contained in the bounding rectangle of $I_\M$, let $\p, \q$ be the top-left and bottom-right corner of $I_R$ respectively, and $\rr,\mathbf{s}$ be the bottom-left and top-right corners respectively. 
Denote the band between two lines $\Delta_\p$ and $\Delta_\q$ passing through $\p$ and $\q$ respectively, as $\Delta_{R}$, and the complement of $\Delta_{R}$ is denoted as $\bar{\Delta}_R$. 
Observe that $d_I(M, R)$ can be calculated as $\di(M, R)=\max\{\di({M|_{\bar{\Delta}_R}}, 0), \di(M|_{\Delta_R}, R|_{\Delta_R})\}$.
Notice that we are using $R|_{\Delta R}$ in the above expression even though 
$R|_{\Delta R}$ has the same underlying interval as $R$, i.e.
$I_{R|_{\Delta R}}=I_R$ since $I_R\subset \Delta_R$.
We use $R|_{\Delta R}$ because $\di$ is only well defined for modules defined on the same poset.
 We claim that
(proof in ~\ref{sec:missing_proofs}):
\begin{proposition}\label{thm:form_di}
    For $I_R$ contained in the bounding rectangle of $I_M$,
    \begin{equation}\label{eq:eq_di}
        \di(M, R)=\max\{\di({M|_{\bar{\Delta}_R}}, 0),  (\di(M|_{\Delta_R}, 0)\vee \di(R|_{\Delta_R}, 0)) \wedge \dH(M|_{\Delta_R}, R|_{\Delta_R})\}.
    \end{equation}
\end{proposition}
The last term can be written explicitly as follows:
\begin{proposition}\label{prop:comp_dH}
    \begin{align*} 
        \dH(M|_{\Delta_R}, R|_{\Delta_R})=\max\{&\dl_\pm(\mathbf{p}, U(I_M)),\dl_\pm(\mathbf{q}, U(I_M)), \dl_\pm(\mathbf{r}, L(I_M)) ,\\ 
        -&\dl_\pm(\mathbf{p}, L(I_M)), -\dl_\pm(\mathbf{q}, L(I_M)), -\dl_\pm(\mathbf{s}, U(I_M))\}.
    \end{align*}
\end{proposition}



\begin{remark*}
    For a rectangle module $R$ which does not satisfy the condition that $I_R$ is contained in the bounding rectangle of $I_M$ as required by Proposition~\ref{thm:form_di}, we show that there exists a rectangle module $R'$ with $I_R'$ contained in the bounding rectangle of $I_M$ such that $d_I(M, R')\leq \di(M, R)$. This claim is stated and proved as Proposition~\ref{prop:opt_rectangle_in_bounding_rectangle} in Appendix. 
\end{remark*}
    
    Based on above observations,
now we give the algorithm to compute $\boxed{M}^*$. Our algorithm searches a finite partition of $\mathbb{R}^2$ to place the four corners of the optimal rectangle while computing the distance between the rectangle and the module in question exactly. 
Denote $I_R=[\rr, \s]$ with $\mathbf{p}, \mathbf{q}$ being its top-left and bottom-right corners respectively. Denote the bounding rectangle of $I_M$ as $B$.
Consider all the lines $\Delta_{\vv}, \vv\in V(I_M)$, which partition $B$ into a collection of bands within $B$, denoted as $\mathcal{D}=\{B_1, B_2,\cdots, B_\ell\}$. 
The steps of our algorithm are as follows:
\begin{enumerate}
    \item Compute a finite partition of $B$ with bands in $\mathcal{D}$ as described above.
    \item Initialize $\boxed{M}^*\leftarrow 0$ to be the zero module.
    \item For each quadruple $i,j,k,l$, under constraints $\mathbf{p}\in {B_i}, \mathbf{q} \in B_{j}, \mathbf{r}\in {B_k}, \mathbf{s}\in {B_l}$, find the optimal rectangle approximation $R^*=\argmin_{I_R=[\mathbf{r}, \mathbf{s}]}{\di(M, R)}$
     according to the Equation~\ref{eq:eq_di} in Proposition~\ref{thm:form_di}. 
    If $\di(M, R^*)\leq \di(M, \boxed{M}^*)$, update $\boxed{M}^*\leftarrow R^*$.
    \item Return $\boxed{M}^*.$
\end{enumerate}

\begin{figure}[ht]
\centerline{\includegraphics[width=0.5\textwidth]{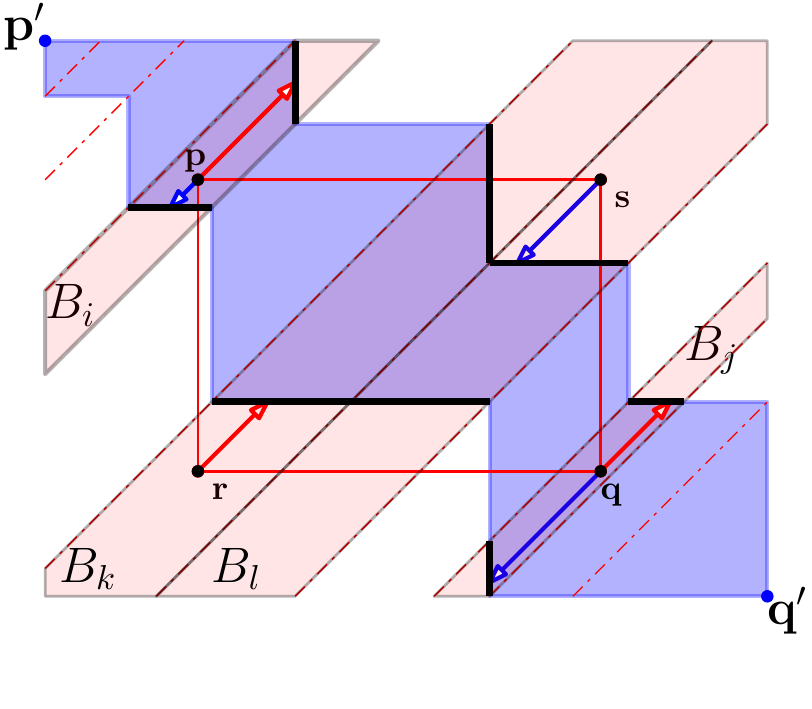}}
\caption{$I_M=$Blue region. Red rectangle, $I_R=[\rr,\s]$, is a candidate rectangle with constraints $\mathbf{p}\in B_i, \mathbf{q} \in B_{j}, \mathbf{r}\in B_k, \mathbf{s}\in B_l$. Six arrowed line segments represent the six terms for computing $\dH(M|_{\Delta_R}, R|_{\Delta_R})$.
}
\label{fig:construction2}
\end{figure}


In each iteration of step 3, the optimization problem can be solved by linear programming since 
all constraints and the terms related to $\di(M, R)$ in Equation~\ref{eq:eq_di} can be expressed as linear expressions on the coordinates of 
$\mathbf{p}=(p_1, p_2)$ and $\mathbf{q}=(q_1, q_2)$  
along with $\max$ and $\min$ operators.
For example, $\di(M|_{\bar{\Delta}_R}, 0)=\max\{\|\Delta_\vv\cap I_M\|_\infty : \vv\in V(I_M\cap \bar{\Delta}_R)\}$ where $\|\Delta_\vv\cap I_M\|_\infty$ is the length of the line segment under infinity norm. Essentially these are the longest line segments coming from the intersections between $I_M\cap \bar{\Delta}_R$ and diagonal lines passing through all vertices of $I_M\cap \bar{\Delta}_R$. 
Note that almost all these lengths $\{\|\Delta_\vv\cap I_M\|_\infty : \vv\in V(I_M\cap \bar{\Delta}_R) \}$ are constant numbers determined by a vertex in $V(I_M)$ except the ones on the boundary of $\bar{\Delta}_R$ which are determined by $\Delta_\p$ and $\Delta_\q$ that may not pass through any vertex in $V(I_M)$. These two line segments on the boundary of $\Delta_R$ can be represented by linear equations in terms of coordinates of $\p=(p_1, p_2)$ and $\q=(q_1, q_2)$. For example, $\|\Delta_\p\cap I_M\|_\infty$ can be represented by $\dl_\pm(\mathbf{p}, U(I_M))-\dl_\pm(\mathbf{p}, L(I_M))$. Since $\p$ is within a band $B_i$ which is determined by two consecutive vertices of $V(I_M)$, the intersections $B_i\cap U(I_M), B_i\cap L(I_M)$ are either horizontal or vertical line segments. If both of them are horizontal or vertical, then $\|\Delta_\p\cap I_M\|_\infty=c$ is a constant number equal to the difference between the coordinates of these two lines. If they are in different directions, say $B_i\cap U(I_M)$ is on the horizontal line $\{y=b\}$ and $B_i\cap L(I_M)$ is on the vertical line $\{x=a\}$, then $\dl_\pm(\mathbf{p}, U(I_M))=b-p_2$ and $\dl_\pm(\mathbf{p}, L(I_M))=a-p_1$. Therefore, $\|\Delta_\p\cap I_M\|_\infty=b-a-(p_2-p_1)$.

To compute $\di({M|_{\bar{\Delta}_R}}, 0)$ and $\di({M|_{\Delta_R}}, 0)$ efficiently, one can pre-compute $\|\Delta_\vv\cap I_M\|_\infty$ for each $\vv\in V(I_M)$ 
and order them according to $v_2-v_1$ where $\vv=(v_1,v_2)$. Then for each pair of $\vv_i, \vv_j\in V(I_M)$, compute and store the values
\begin{eqnarray*}
\diam(i,j)\triangleq\max{\{\|\Delta_{\vv_k}\cap I_M\|_\infty: i\leq k\leq j\}}\mbox{ and}\\ \overline{\diam}(i,j)\triangleq\max{\{\|\Delta_{\vv_\ell}\cap I_M\|_\infty:  \ell<i \mbox{ or } \ell>j\}}. 
\end{eqnarray*}
This can be done in $O(|V(I_M)|^2)$ time.
For a given $\Delta_R$ between boundaries $\Delta_\p$ and $\Delta_\q$, we have $\vv_i, \vv_j\in V(I_M)\cap \Delta_R$ such that $\Delta_\p$ is within the band $\Delta_{\vv_i}, \Delta_{\vv_{i-1}}$ and $\Delta_\q$ is within the band $\Delta_{\vv_j}, \Delta_{\vv_{j+1}}$. 
See Figure~\ref{fig:di_m_0_illustrate} for an illustration.
Then we have 
\begin{eqnarray*}
\di({M|_{\bar{\Delta}_R}}, 0)=\|\Delta_\p\cap I_M\|_\infty\vee\|\Delta_\q\cap I_M\|_\infty\vee\diam(i,j) \mbox{ and}\\ \di({M|_{\Delta_R}}, 0)=\|\Delta_\p\cap I_M\|_\infty \vee\|\Delta_\q\cap I_M\|_\infty\vee\overline{\diam}(i,j). 
\end{eqnarray*}
Since we pre-compute and store $\diam(i,j)$ and $\overline{\diam}(i,j)$, each iteration incurs only $O(1)$ time for accessing constant number of them.
\begin{figure}[ht]
\centerline{\includegraphics[width=0.5\textwidth]{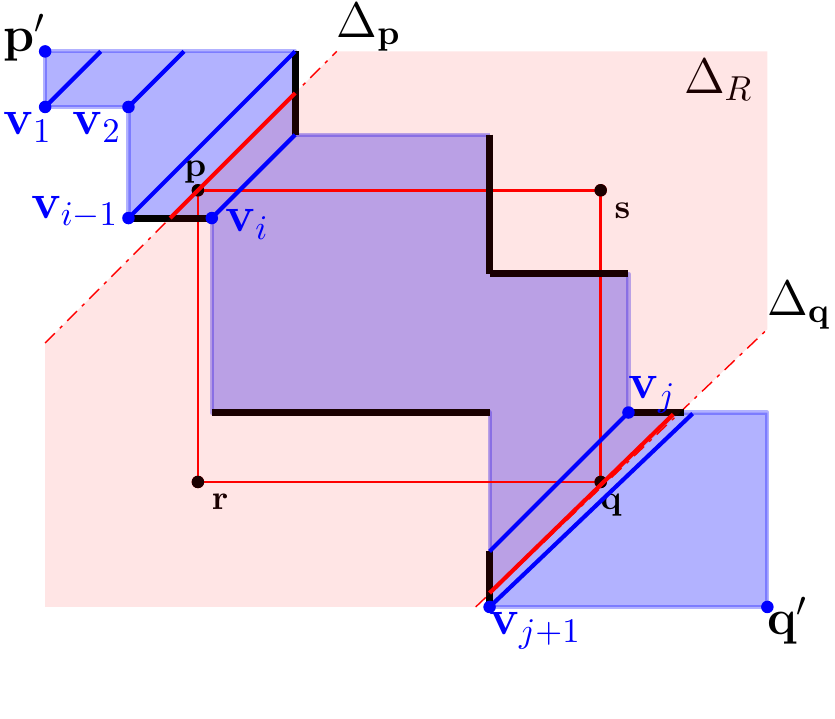}}
\caption{The blue diagonal line segments are $\Delta_{\vv}\cap I_M$, and the red diagonal line segments are $\Delta_{\p}\cap I_M$ and $\Delta_{\q}\cap I_M$.}
\label{fig:di_m_0_illustrate}
\end{figure}

{\bf Complexity.}
Step 3 in the algorithm takes constant time. The last term $\dH(M|_{\Delta_R}, R|_{\Delta_R})$ involves constant number of linear equations. The third term $\di(R|_{\Delta_R}, 0)=\min\{q_1-p_1, p_2-q_2\}$ can be computed in constant time. First two terms $\di(M|_{\bar{\Delta}_R}, 0)$ and $\di(M|_{\Delta_R}, 0)$ can also be evaluated
in constant time if we pre-compute the lengths of a collection of line segments as described above. In particular, only two line segments one on each $\Delta_\p$ and $\Delta_\q$ relate to  optimizing variables $\p$ and $\q$ and hence enter into our optimization which can be
solved in $O(1)$ time. To see this, observe that
the feasible region of $\di(M, R)$ is determined by a constant number of linear equations. An optimal solution is then given by a corner point on the boundary of the feasible region. With a brute-force computation, we can determine all these corner points and then take the optimal one in $O(1)$ time. One can also divide the entire optimization problem into several sub-problems with different feasible regions such that each sub-problem is a min-max problem which can be solved by linear programming, and the optimal solution becomes the minimum over all the solutions of sub-problems. The time complexity still remains $O(1)$ though this division into sub-problems can be more efficient in practice.

With $|\mathcal{D}|=O(|V(I_M)|)$ being the total number of vertices of the input interval $I$, the time complexity of the whole algorithm is $O(|V(I_M)|^4)$.

Now we can compute $\boxed{M}^*$ for any interval decomposable module $M=\oplus M_i$ by computing $\boxed{M_i}^*$ for each $i$. We have a result similar to Proposition~\ref{prop:db_bound0} 
but with a tighter bound on the additional $\epsilon$ terms.

\begin{theorem}\label{thm:db_bound1}
For two interval decomposable modules $M$ and $N$, 
\begin{equation*}
\di(M,N)\leq \db(M,N) \leq \db(\boxed{M}^*, \boxed{N}^*) + \epsilon_M^* + \epsilon_N^* \leq 3\di(M, N) + 4(\epsilon_M^* + \epsilon_N^*)
\end{equation*}
where $\epsilon_M^*=\db(M, \boxed{M}^*)$ and $\epsilon_N^*=\db(N, \boxed{N}^*)$.
\end{theorem}
\begin{remark}
When $M, N$ are rectangle decomposable modules, $\epsilon_M^*=\epsilon_N^*=0$. Therefore, our result is a generalization of Bjerkevik's~\cite{bjerkevik2021stability} stability theorem~\ref{thm:rectangle_stability}.

From this theorem we can see that $\frac{1}{3}\db(M,N) - \frac{4}{3} (\epsilon_M^* + \epsilon_N^*) \leq \di(M,N)\leq \db(M,N)$.
In general, it can be hard to approximate $\di$ by a constant factor. Bjerkevik et al.~\cite{Bjerkevik2020} shows that approximating interleaving distance within a constant factor less than 3 is NP-hard. This additional $\epsilon$ terms may be inevitable in which case this may turn out to be a good polynomial-time computable approximation on $\di$. Also, this bound points to where the gap between $\db$ and $\di$ comes from. These $\epsilon$ terms depend on the structures of $M$ and $N$ themselves, which measure how far an interval decomposable module is from a rectangle decomposable module. 

\end{remark}

\section{Generalized matching distance}\label{sec:general_approx}
In this section, we give an application of our rectangle approximation to define a distance on general 2-parameter persistence modules which can be viewed as a generalization of 
matching distance~\cite{biasotti2011new,landi2014rank}. 
As we know, matching distance and bottleneck distance are most commonly used as a lower and an upper bound respectively for interleaving distance. 
Following~\cite{landi2014rank}, matching distance $\dmatch$ is defined as follows.
\begin{definition}\label{def:dmatch}
For each line $\mathbf{L}$ in $\Real^2$, use $\mathbf{L}(t)=t\mathbf{a}+\mathbf{b}$ 
with $a_1\vee a_2=1, a_1\wedge a_2>0, b_1+b_2=0$ as its 1-parameterization over $t\in \Real$. For a module $M$, define a 1-parameter persistence module  $M_\mathbf{L}(t)\triangleq M(\mathbf{L}(t))$. Then
\begin{equation*}
    \dmatch(M, N)\triangleq\sup_{\mathbf{L}:\mathbf{L}(t)=t\mathbf{a}+\mathbf{b}}(a_1\wedge a_2)\db(M_{\mathbf{L}}, N_\mathbf{L}).
\end{equation*}
\end{definition}


Unfortunately, both matching and bottleneck distances have some drawbacks. 
Although matching distance can be computed (approximately or exactly~\cite{kerber2018exact_compute_matching_distance}) in polynomial time, as indicated in~\cite{cerri2016coherent, cerri2019coherent_geometrical}, it forgets the natural association between the homological properties of slices that are close to each other. On the other hand, the bottleneck distance emphasizes too much on the natural association among slices, which makes the matching obtained by bottleneck distance coherent within each indecomposable component but not stable with respect to the interleaving distance that considers the entire module.

In this section, we propose a general strategy to define a distance $\dbmm^\CC$ depending on a chosen convex covering $\CC$ on $\Real^2$ (defined later) which captures the natural association of slices like bottleneck distance but still enjoys some stability-like property based on our previous results as follows
\begin{equation}
    \dmatch(M, N)\leq \dbmm^{\CC}(M,N)\leq \min\{3\di(M, N)+\epsilon^*(M, N),\, \db(M, N)\}
\end{equation}
where $\epsilon^*(M, N)$ is determined by the covering $\CC$ as
\begin{equation}\label{eq:epsilon_star_def}
\epsilon^*(M, N)=4\cdot\sup_{C\in \CC, \mathbf{a}:a_1\wedge a_2=1}\db(M^\aaa|_{C}, \boxed{M^\aaa|_{C}}^*)+\db(N^\aaa|_{C}, \boxed{N^\aaa|_{C}}^*).
\end{equation}
Here $M^\aaa$ is a scaled module of $M$ defined as 

\begin{definition}\label{def:scaled_module}
    $M^\aaa(\uu)=M(\uu\odot\aaa)$ where $\uu\odot\aaa=(u_1a_1, u_2a_2)$ known to be the Hadamard product.
     See Figure~\ref{fig:scaled} for an illustration.
\end{definition}

By carefully choosing the covering $\CC_\alpha$ for any $\alpha\in[0,1]$, we can 
further get a distance $\dbmm^{\CC_\alpha}$ that enjoys $(3+\alpha)$-stability 
as follows:
\begin{equation}
    \dmatch(M, N)\leq \dbmm^{\CC_\alpha}(M,N)\leq \min\{(3+\alpha)\di(M, N), \; \db(M, N)\}.
\end{equation}
We will show that the matching distance can be viewed as a special case of our construction when $\CC$ is chosen to be the set containing all slope-1 lines in $\Real^2$. 
The main idea is to cover the entire plane $\Real^2$ by some collection of bands $\CC$. By carefully choosing these bands we make sure that the restricted modules $M|_C$ and $N|_C$ within each band $C\in \CC$ are interval decomposable. Then, locally we can compute the bottleneck distance which bounds the interleaving distance in each band according to our previous results. 
As in~\cite{cerri2011new}, our proposed distance can be approximated by choosing a finite set of directions. Recently, in~\cite{kerber2018exact_compute_matching_distance}, a polynomial time algorithm has been presented for computing the matching distance exactly. It remains open
if our proposed distance can also be computed exactly using similar ideas. 



We call a set $\CC$ of convex sets a convex covering of $\Real^2$ 
if the union of the sets of $\CC$ equals $\Real^2$.
For any convex set $C\subseteq \Real^n$, we have induced $C$-modules $M|_C$ and $N|_C$. Therefore, we can define $\di(M|_C, N|_C)$ with the following property:
\begin{fact}\label{fact:subset_convex_interleaving}
For two convex sets $C$ and $C'$, if $C'\subseteq C$, then $\di(M|_{C'}, N|_{C'} )\leq \di(M|_C, N|_C)$.
\end{fact}
We need the following definitions for further development.
\begin{definition}
Define $\PAR(\Real^2)\triangleq\{\CC\mid \CC \mbox{ is a convex covering on }\Real^2 \}$. For two convex coverings $\CC \mbox{ and } \CC'$, we say $\CC'$ is a refinement of $\CC$, denoted as $\CC'\subseteq \CC$,  if $\forall C'\in \CC', \exists C\in \CC, C'\subseteq C$.
For $\CC\in \PAR(\Real^2)$, define
$\di^\mathcal{C} (M,N)\triangleq \sup_{C\in \mathcal{C}} \di (M|_C, N|_C)$ and $\db^\mathcal{C} (M,N)\triangleq \sup_{C\in \mathcal{C}} \db (M|_C, N|_C)$ . 
\end{definition}

\begin{remark}
The refinement relation defines a partial order on $\PAR(\Real^2)$. The minumum covering is the discrete covering consisting of all singleton sets of $\Real^2$. The maximum covering is $\{\Real^2\}$. Furthermore, $\di^{(\cdot)}$ can be viewed as an order preserving mapping from poset $\PAR(\Real^2)$ to pseudometric functions on bimodules. That is, 
\end{remark}

\begin{proposition}\label{prop:order_preserving_function_on_PAR}
    For $\CC'\leq \CC\in \PAR(\Real^2)$, $ d^\mathcal{C'}_I \leq  d^\mathcal{C}_I$. Specifically, $\di^{\{\Real^2\}}=\di$ and $\db^{\{\Real^2\}}=\db$.
\end{proposition}


Most of $\di^\CC$ are hard to analyze in general which prompts us
to construct some special covering $\CC$.
Consider a finite subset $\SSet=\{\s^{(1)}, \s^{(2)}, \cdots, \s^{(l)}\}\subset\Real^2$ ordered by $\Delta_{\s^{(i)}}$, which means $\s^{(i)}\leq \s^{(j)}$ if $s^{(i)}_2 - s^{(i)}_1\leq s^{(j)}_2-s^{(j)}_1$. All these $\Delta_{\s^{(i)}}$ partition $\Real^2$ into finitely many regions. We collect all these 
regions with boundaries to construct the covering $\CC_\SSet$ for the set $\SSet$. For $\SSet=\emptyset$, we define $\CC_{\emptyset}=\{\Real^2\}$. We extend our notation to incorporate
the special set $\SSet=\Real^2$ by defining $\CC_{\Real^2}=\{\Delta_\s \mid \s\in \Real^2\}$, which is the collection of all diagonal lines in $\Real^2$. By Proposition~\ref{prop:order_preserving_function_on_PAR}, we immediately have 
\begin{proposition}
    For any finite set $\SSet\subset \Real^2$, $\CC_{\Real^2}\leq \CC_{\SSet} \leq \CC_\emptyset$. 
\end{proposition}

Before we define our generalized matching distance, we need another notation. Let $\mathbf{A}=\{\aaa=(a_1, a_2)\in\Real^2_+\mid a_1\wedge a_2=1 \}$.
For a vector $\mathbf{a}=(a_1, a_2)\in \mathbf{A}$, recall that the scaled module $M^{\aaa}$ of $M$ is $M^{\aaa}(\mathbf{u})\triangleq M(\mathbf{a}\odot\mathbf{u})=M(a_1\cdot u_1, a_2\cdot u_2)$ (Figure~\ref{fig:scaled}). 
Note that $M^{(1,1)}=M$. Similarly, we also define the scaled covering $\CC_\SSet^\aaa\triangleq \CC_{\aaa^{-1}\odot \SSet}$ for a covering $\CC_\SSet$, where ${\aaa^{-1}\odot \SSet}=\{\aaa^{-1} \odot \s\mid \s\in \SSet\}$. Intuitively, scaled modules and scaled coverings together account for
all directions $\mathbf{a}\in \mathbf{A}$ needed to be considered for matching distance even though
coverings align along the diagonal direction.

\begin{figure}[ht]
\centerline{\includegraphics[width=0.5\textwidth]{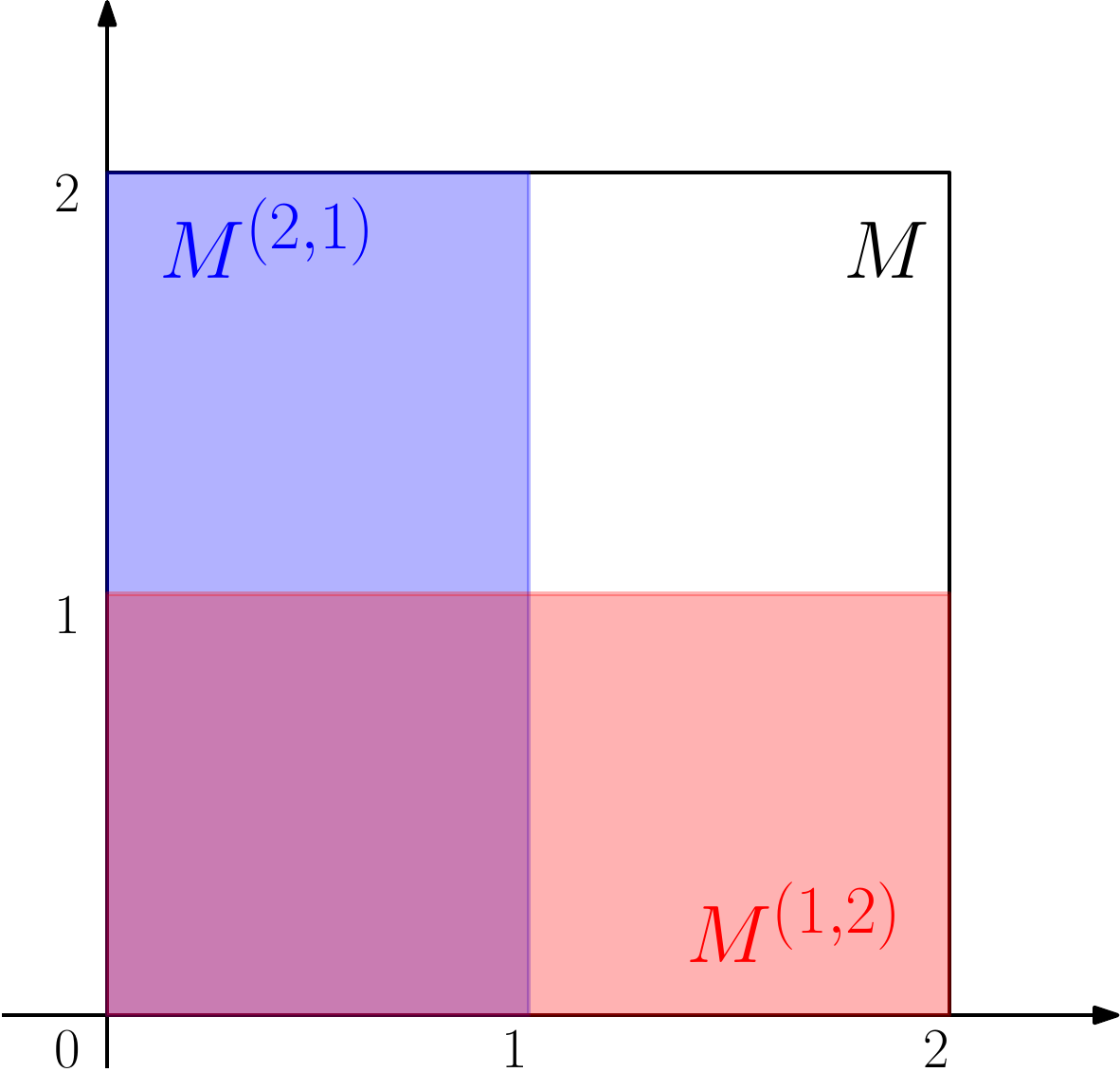}}
\caption{$M$ is a module on $[0,2]\times [0,2]$,
the blue region and red region are the underlying intervals of $M^{(2,1)}$ and $M^{(1,2)}$ respectively.}
\label{fig:scaled}
\end{figure}

Now we construct the generalized matching distance.
\begin{definition}
    Given a covering $\CC_\SSet$, for two persistence modules $M, N$, define 
\begin{equation*}
    \tilde{\di}^{\CC_\SSet}(M, N)\triangleq \sup_{\aaa\in\mathbf{A}}\di^{\CC_\SSet^{\aaa}}(M^\aaa, N^\aaa) = \sup_{\aaa\in\mathbf{A}}\sup_{C\in\CC_\SSet^{\aaa}}\di(M^\aaa|_C, N^\aaa|_C)
\end{equation*}
 and 
\begin{equation*}
    \tilde{\db}^{\CC_\SSet}(M, N)\triangleq \sup_{\aaa\in\mathbf{A}}\db^{\CC_\SSet^{\aaa}}(M^\aaa, N^\aaa) = \sup_{\aaa\in\mathbf{A}}\sup_{C\in\CC_\SSet^{\aaa}}\db(M^\aaa|_C, N^\aaa|_C).
\end{equation*}
\end{definition}
Based on our definition, we have the following properties:
\begin{proposition}\label{prop:tdi=di}
\begin{enumerate*}
    \item For $\CC_{\SSet'}\leq \CC_\SSet$, $\tilde{\di}^{\CC_{\SSet'}}\leq \tilde{\di}^{\CC_\SSet}$ and $\tilde{\db}^{\CC_{\SSet'}}\leq \tilde{\db}^{\CC_\SSet}$;\\
    \item $\tdi^{\CC_{\Real^2}}=\tdb^{\CC_{\Real^2}}=\dmatch$; \\
    \item $\tdi^{\CC_\emptyset}=\di, \tdb^{\CC_\emptyset}=\db$.\\
\end{enumerate*}
\end{proposition}
\begin{proof}
    The first one is an immediate corollary from Proposition~\ref{prop:order_preserving_function_on_PAR}.
    Consider the second one. 
    The first equality comes from isometry theorem of 1-parameter persistence module. To show the second equality, recall Definition~\ref{def:dmatch} of $\dmatch$:
    \begin{equation*}
        \dmatch(M, N)\triangleq\sup_{\substack{\mathbf{L}:\mathbf{L}(t)=t\mathbf{a}+\mathbf{b},\\ a_1\vee a_2=1,\\ a_1\wedge a_2>0, \\b_1+b_2=0}}(a_1\wedge a_2)\db(M_{\mathbf{L}}, N_\LL)
    \end{equation*}
    where $M_{\mathbf{L}}$ is defined by $M_{\mathbf{L}}(t) = M(\mathbf{L}(t))$.

    First, note that, in the definition of $\dmatch$, we can remove the scalar weight term $(a_1\wedge a_2)$ by changing parameterization of $\LL(t)$ with $a_1\wedge a_2=1, b_1+b_2=0$.
    Now we obtain an equivalent definition of $\dmatch$ as follows.
    \begin{equation*}
        \dmatch(M, N)=\sup_{\substack{\mathbf{L}:\mathbf{L}(t)=t\mathbf{a}+\mathbf{b},\\ a_1\wedge a_2=1, \\b_1+b_2=0}}\db(M_{\mathbf{L}}, N_\mathbf{L}).
    \end{equation*}
    
    For each $\LL$, let $\LL'\triangleq\aaa^{-1}\odot\LL$ given by $\LL'(t)=t(1,1)+\aaa^{-1}\odot\mathbf{b}$.
    Observe that
    \begin{equation*}
    M^\aaa|_{\LL'}(\LL'(t))=M(\aaa\odot\LL'(t))=M(\aaa\odot\aaa^{-1}\odot\LL(t))=M(\LL(t))=M_\LL(t).
    \end{equation*}
    This means $\forall\LL:\LL(t)=t\aaa+\mathbf{b}$, the 1-parameter module $M_\LL$ and the corresponding $\aaa$-scaled model $\LL'$-module $M^{\aaa}|_{\LL'}$ are the same in the sense that $\forall t, M_\LL(t)=M^\aaa|_{\LL'}(\LL'(t))$. It is not difficult to check that the corresponding 
    linear maps are also the same. 
    Taking the supremum over all lines $\LL$, we get $\dmatch=\tdb^{\CC_{\Real^2}}$ by definition.
    
    For the third one, the second equality is an immediate result from the first equality. 
    For the first equality, 
    we need to show that $\di(M^{\aaa}, N^{\aaa})\leq\di(M, N), \forall\mathbf{a}$. Assume $M, N$ are $\delta$-interleaved with interleaving morphisms $\phi:M\rightarrow N_{\rightarrow \vec{\delta}}$ and $\psi:N\rightarrow M_{\rightarrow \vec{\delta}}$. We claim that $M^{\aaa}, N^{\aaa}$ 
    are $\delta$-interleaved witnessed by maps
    $\phi^{\aaa}$ and $\psi^{\aaa}$ between $M^{\aaa}$ and  $N^{\aaa}$ constructed as follows. 
    \begin{flalign*}
        \phi^{\aaa}(\x):=&N(\aaa\odot\x+\vec{\delta}\rightarrow \aaa\odot\x+\aaa\odot\vec{\delta})\circ\phi(\aaa\odot\x) \\
        \psi^{\aaa}(\x):=&M(\aaa\odot\x+\vec{\delta}\rightarrow \aaa\odot\x+\aaa\odot\vec{\delta})\circ\psi(\aaa\odot\x)
    \end{flalign*}
    Notice that $\aaa$ satisfies $a_1\wedge a_2=1$. So, $\aaa\odot\x+\vec{\delta}\leq \aaa\odot\x+\aaa\odot\vec{\delta}$. Therefore, all the arrows `$\rightarrow$' above are well-defined.
    The construction for $\phi^{\aaa}, \psi^{\aaa}$ are symmetric. 
    By commutative properties of $\phi, \psi$ and linear maps within $M$ and $N$, 
    it can be verified that  
    $\phi^{\aaa}, \psi^{\aaa}$ 
    indeed provide $\delta$-interleaving between 
    $M^{\aaa}$ and $N^{\aaa}$. We 
    leave details in~\ref{sec:missing_proofs2}.
\end{proof}

Therefore, $\dmatch$ can be viewed as a special case which bounds $\tdi$ and $\tdb$ from below.
Computing 
$\tdi$ or $\tilde{\db}$ in general can be intractable. 
For some $\CC_{\SSet}$ constructed in a specific way, we provide an efficient algorithm to approximate $\tdi^{\CC_{\SSet}}$ by
a finite subset $\hat{\mathbf{A}}\subseteq \mathbf{A}$.
The idea is to construct a set $\SSet$ for $M, N$ so that for any $\aaa\in \mathbf{A}$ and any $C\in \CC_\SSet^\aaa$, both $M^\aaa|_C$ and $N^\aaa|_C$ are interval decomposable. Based on such a set $\SSet$,
we can approximate $\tdb$ efficiently and utilize the results in the previous section to bound $\tdi$.
Before we give the construction of the set $\SSet$, we first introduce the results and the algorithm.

\begin{proposition}\label{prop:tdb_stability}
    Given two modules $M, N$ and a covering $\CC_\SSet$ such that $\forall\aaa\in \mathbf{A}, \forall C\in \CC_\SSet^\aaa$, both $M^\aaa|_C$ and $N^\aaa|_C$ are interval decomposable, we have
    \begin{align}
        \dmatch(M, N) \leq \tdi^{\CC_\SSet}(M, N)\leq \tdb^{\CC_\SSet}(M, N) &\leq \min\{3\tdi^{\CC_\SSet}(M, N)+\epsilon^*(M, N), \db(M, N) \}\\
        &\leq  \min\{3\di(M, N)+\epsilon^{\CC_\SSet}(M, N), \db(M, N) \}
    \end{align}
where 
\begin{equation} 
\epsilon^{\CC_\SSet}(M, N)=4\cdot\sup_{C\in \CC_\SSet, \mathbf{a}\in\mathbf{A}}\db(M^\aaa|_{C}, \boxed{M^\aaa|_{C}}^*)+\db(N^\aaa|_{C}, \boxed{N^\aaa|_{C}}^*).
\end{equation}
\end{proposition}

The distance $\tdb^{\CC_\SSet}$ is defined as a supremum over all scaled modules $M^\aaa$ and $N^\aaa$ for every $\aaa\in \mathbf{A}$. 
To develop a polynomial time approximation algorithm, following~\cite{cerri2011new}, we take a finite subset $\hat{\mathbf{A}}\subseteq\mathbf{A}$ to approximate the distance.
We summarize the algorithm for approximating $\tdb$:
\begin{enumerate}
    \item Compute the minimal presentations of $M$ and $N$ by the algorithm given in~\cite{lesnick2019computing_minimal_presentation}.
    \item Compute the decomposition $M\simeq \oplus M_i$ and $N\simeq \oplus N_j$ by the algorithm given in~\cite{DeyCheng19}.
    \item For every non-interval component $M_i$ and $N_j$, select a set $\SSet\subseteq \mathbb{R}^2$ satisfying the premise of Proposition~\ref{prop:tdb_stability} and construct $\CC_\SSet$. Take the union over all these sets of $\SSet$'s.
    \item For $\forall \aaa\in\hat{\mathbf{A}}, \forall C\in\CC_\SSet^{\mathbf{a}}$, compute the bottleneck distance $\db(M^{\mathbf{a}}|_{C}, N^\mathbf{a}|_{C})$ by the algorithm given in~\cite{DX18}. 
    \item Output $\displaystyle\max_{\aaa\in \hat{\mathbf{A}}, \Delta_C\in\Delta_{\CC_\SSet^\aaa}} \db(M^{\mathbf{a}}|_{C}, N^\mathbf{a}|_{C})$ as an approximation of $\tdb^{\CC_\SSet}(M, N)$.

\end{enumerate}

The upper bound in Proposition~\ref{prop:tdb_stability} can be improved for some special covering.
\begin{corollary}
    Let $\CC_\SSet$ be a covering satisfying the condition in Proposition~\ref{prop:tdb_stability} and let $\CC_{\SSet'}\subset \CC_\SSet$ be a refinement such that
    $\epsilon^{\CC_{\SSet'}}(M, N)\leq \alpha \dmatch(M, N)$,
    then we have
\begin{equation}
    \dmatch(M, N) \leq \tdb^{\CC_{\SSet'}}(M, N) \leq \min\{(3+\alpha)\cdot \di(M, N),\,\db(M, N)\}.
\end{equation}    
\end{corollary}

Such a refinement $\CC_{\SSet'}$ can simply be obtained as follows. Each $C\in\CC_\SSet$, which does not satisfy the condition that $4\sup_{\mathbf{a}\in\mathbf{A}}[\db(M^\aaa|_{C}, \boxed{M^\aaa|_ {C}}^*)+\db(N^\aaa|_{C}, \boxed{N^\aaa|_{C}}^*)] \leq \alpha\dmatch(M, N)$, is subdivided evenly into several bands $C_i$'s so that $2\epsilon_{C_i}\leq \alpha\dmatch(M, N)$. Thanks to Proposition~\ref{prop:subdivide_bound} in~\ref{sec:missing_proofs2}, we 
get that
\begin{equation*}
4\sup_{C_i}\sup_{\mathbf{a}\in\mathbf{A}}\db(M^\aaa|_{C_i}, \boxed{M^\aaa|_ {C_i}}^*)+\db(N^\aaa|_{C_i}, \boxed{N^\aaa|_{C_i}}^*) \leq \alpha\dmatch(M, N).
\end{equation*}
Now $\CC_{\SSet'}$ as constructed above satisfies that $\epsilon^{\CC_{\SSet'}}(M, N)\leq \alpha \dmatch(M, N)$. 
If $C$ is an infinity band (not bounded by two lines), 
one can subdivide $C$ starting
from the lower or upper boundary iteratively until the condition is satisfied. 

\begin{remark*}
    From above constructions, we can see that the choice of covering $\CC$ gives a trade-off between lower and upper bounds. Finer construction gives us a distance closer to the matching distance $\dmatch$, while coarser construction provides a distance which is closer to the bottleneck distance $\db$.
\end{remark*}

Now we discuss how to construct the set $\SSet$ in step 3 which should satisfy the premise of Proposition~\ref{prop:tdb_stability}.
For this we need the idea of finite presentation which is well studied in commutative algebra 
for structures like graded modules.
Here we follow the description in~\cite{kerber2018exact_compute_matching_distance} to give a more
intuitive definition of finite presentation as a graded matrix which does not require too much background in commutative algebra (one can check~\cite{atiyah2018introduction, cox2006usingAG, lesnick2015theory} for more details about finite presentation).

For any band described before $C\subset \Real^2$, a finite presentation is an $n\times m$ matrix $P$ over field $\field{F}_2$,
where each row $i$ is labeled by a point $\gr(\row_i)\in C$, and each column $j$ is labeled by a point $\gr(\col_j)\in C$, called the grades, such that $P_{i,j}=1\implies \gr(\row_i) \leq \gr(\col_j)$.
We denote all grades of rows as $\gr(\row P)$ and all grades of columns as $\gr(\col P)$. The grades $\gr(\row P)$ and $\gr(\col P)$ represent generators and relations of a persistence module constructed as follows:
Let $e_1, e_2, \cdots, e_l$ denote the standard basis of $\field{F}_2^l$. 
For any $\mathbf{u}\in {C}$, define the subset 
$\gen(\mathbf{u})\triangleq\{e_i\mid \gr(\row_i)\leq \mathbf{u}\}$ and $\rel(\mathbf{u})\triangleq\{\col_j\mid \gr(\col_j)\leq \mathbf{u}\}$. 
 Then define $M^P(\mathbf{u})=\spanning(\gen(\mathbf{u}))/\spanning(\rel(\mathbf{u}))$, and $M^P(\mathbf{u}\rightarrow \mathbf{v})$ as the map induced by the inclusion map $\spanning(\gen(\mathbf{u}))\rightarrow \spanning(\gen(\mathbf{v}))$. It can be checked that $M^P$ is a $C$- module. If $M^P$ is isomorphic to some persistence module $M$, we say $P$ is a {\emph{presentation matrix}} of $M$. 
\begin{remark*}
Note that our definition above is consistent with the one in~\cite{kerber2018exact_compute_matching_distance} for $C=\Real^2$. We extend it for all band posets we care about here. 
In practice, 
the data is often given as a simplicial filtration which induces the persistence module. It is shown in~\cite{lesnick2019computing_minimal_presentation} that for a $2$-parameter persistence module, a finite presentation can be computed in cubic time. 
\end{remark*}






Note that $\gr(\col P)$ and $\gr(\row P)$ are two subsets of $\Real^2$ with partial order naturally induced from $\Real^2$. Specifically, we claim the following Proposition for an important special case (proof in \ref{sec:missing_proofs2})
\begin{proposition}\label{prop:linearly_orderd_interval_decomp}
If $P$ is a presentation of $M$ such that neither $\gr(\col P)$ nor $\gr(\row P)$ has incomparable pairs, 
then $M$ is interval decomposable. 
\end{proposition}
Observe that the Proposition above provides only a sufficient condition for interval decomposable
modules. It is not a necessary condition.

Borrowing an idea from~\cite{lesnick2015interactive_visualization_2d_modules}, for any persistence module $M$ with finite presentation $P$, and for any band $C$ with upper and lower boundary $\Delta_a$ and $\Delta_b$ respectively for some $-\infty\leq a\leq b\leq +\infty$ ($a=-\infty$ means  $C$ has no lower boundary, $b=+\infty$ means $C$ has no upper boundary), define a $\push$ operation on any point $\uu\in \Real^2$ to band $C$ as:


\begin{equation}\label{eq:push_B_formula}
\push_C(\mathbf{u}=(u_1,u_2))\triangleq\min\{\ww\in C\mid\ww\leq \uu\}=
    \begin{cases}
     (u_2-b, u_2) \mbox{ if } \uu \mbox{ is above }\Delta_b, \\
     (u_1, u_1+a) \mbox{ if } \uu \mbox{ is below } \Delta_a,\\
     \uu \mbox{ if } \uu\in C.
    \end{cases}
\end{equation}

Intuitively, $\push_C(\mathbf{u})$ is the lower-left corner of the intersection between $C$ and the unbounded rectangle $[\uu, (+\infty, +\infty)]$. See Figure~\ref{fig:push_B} as an illustration.

\begin{figure}[ht]
\centerline{\includegraphics[width=0.5\textwidth]{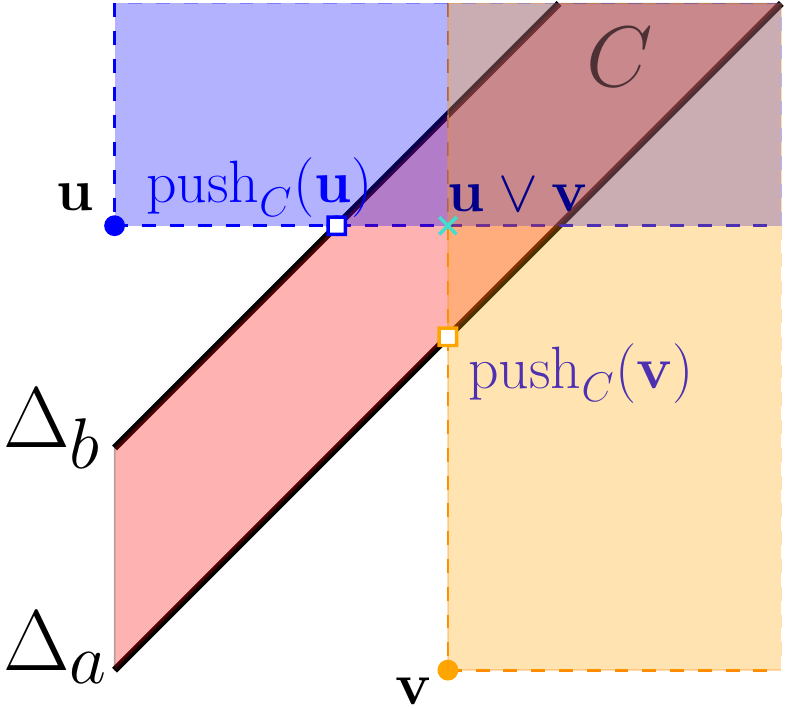}}
\caption{Examples of $\push_C$ of two points $\uu,\vv$. The square points are $\push_C(\uu),\push_C(\vv)$.
}
\label{fig:push_B}
\end{figure}

Now we claim the following Proposition to connect the finite presentation between $M$ and $M|_C$ (proof in \ref{sec:missing_proofs2}).


\begin{proposition}\label{lm:presentation_on_band}
Given a presentation of $M$ with finite presentation $P$ and a band $C$,
the restriction $M|_C$ has a finite presentation $P$ which can be obtained by replacing each $\mathbf{u}$ of $P$ with $\push_C(\mathbf{u})\triangleq\min\{\ww\in C\mid\ww\leq \uu\}$.
\end{proposition}

Now given persistence modules $M$ and $N$ with finite presentations $P_M$ and $P_N$ respectively, let $Gr=\{\gr(\row P_M), \gr(\col P_M), \gr(\row P_N), \gr(\col P_N)\}$.
Let $\SSet=\{\uu\vee\vv\triangleq(u_2\vee v_2, u_1\vee v_1)\mid  \mathbf{u}, \mathbf{v} \mbox{ are incomparable in } G \mbox{ for } G\in Gr\}$ be the anchor set. As an immediate result from Proposition~\ref{prop:linearly_orderd_interval_decomp}, we have the following property which satisfies the requirement in Proposition~\ref{prop:tdb_stability}.



\begin{proposition}\label{prop:B1_interval_decomposable}
Given a persistent module $M$ with finite presentation $P$, let 
$\SSet\triangleq\{\uu\vee\vv\triangleq(u_2\vee v_2, u_1\vee v_1)\mid  \mathbf{u}, \mathbf{v} \mbox{ are incomparable in } \gr(\row P_M) \mbox{ or } \gr(\col P)\}$, 
then we have $\forall \aaa, C\in \CC_\SSet$, 
$M^{\aaa}|_{C}$ is interval decomposable.
\end{proposition}

\section{Conclusion}
In this paper, first we consider interval decomposable modules and their optimal approximations with rectangle decomposable modules with respect to the bottleneck distance. We present a polynomial time algorithm for computing this optimal approximation exactly which, together with the polynomial-time computable bottleneck distance among interval decomposable modules~\cite{DX18}, provides a lower bound on the interleaving distance.
Next, we propose a distance between $2$-parameter persistence modules that is more discriminating than the matching distance and is bounded from above by the bottleneck distance. This distance
can be approximated by a polynomial time algorithm. Using an approach in~\cite{kerber2018exact_compute_matching_distance}, it may be plausible
to compute it exactly--a question we leave open for future research. This distance can bound the interleaving distance from below with some subtractive factors derived from our optimal approximation of interval decomposable modules with rectangle decomposable modules. For this, we propose
a non-trivial covering of a given module 
by interval decomposable modules which can be an interesting technique on its own.

\bibliography{ref}

\appendix

\newpage
\section{Missing proofs in Section~\ref{sec:rectangle_approx}}\label{sec:missing_proofs}
First, we introduce some definitions and results used later in our proof. Most of them comes from~\cite{DX18}.

\begin{definition} [Valid intersection]
An intersection component $Q_i$  is $(M,N)$-\emph{valid} (or \emph{valid} in short if the order of the modules is clear from the context) if for each $\x\in I_{Q_i}$ the following two conditions hold (see Figure~\ref{fig:valid_intersection}):
\[
\mbox{(i) } \y  \leq \x \mbox{ and } \y \in I_M \implies \y \in I_N, \mbox{ and (ii) } \z \geq \x \mbox{ and } \z \in I_N \implies \z \in I_M.
\]
\end{definition}

\begin{figure}[ht]
\centerline{\includegraphics[height=3.5cm]{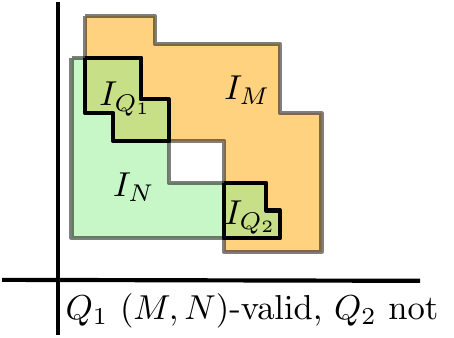}}
\caption{Examples of a valid intersection and an invalid intersection.}
\label{fig:valid_intersection}
\end{figure}

\begin{figure}[ht]
    \centering
    \includegraphics{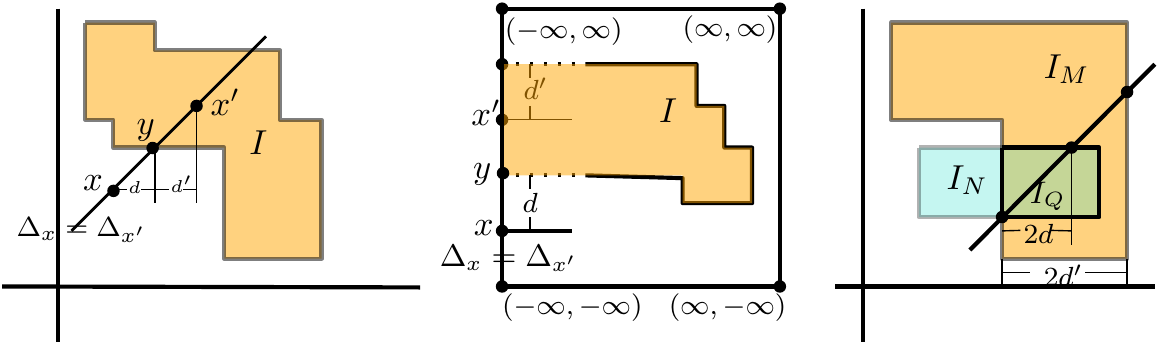}
    \caption{$d=\dl(x,I)$, $y=\pi_I(x)$, $d'=\dl(x',L(I))$ (left); $d=\dl(x,I)$ and $d'=\dl(x',U(I))$ are defined on the left edge of $B(\RB^2)$ (middle); $Q$ is $d'_{(M,N)}$- and $d_{(N,M)}$-trivializable (right). }
    \label{fig:dl}
\end{figure}

 \begin{definition} [Trivializable intersection]
  Let $Q$ be a connected component of the intersection of two modules $M$ and $N$. 
  For each point $\x\in I_Q$, define
  \[
  d^{(M,N)}_{triv}(\x) =\max\{\dl(\x, U(I_M))/2, \dl(\x, L(I_N))/2)\}.
  \]
  For $\delta \geq 0$, we say a point $\x$ is $\delta_{(M, N)}$\emph{-trivializable} if $d^{(M,N)}_{triv}(\x) < \delta$.
  We say an intersection component $Q$ is $\delta_{(M,N)}$\emph{-trivializable} (or \emph{trivializable} for brief if the order of modules is clear from context)  if each point in $I_Q$ is $\delta_{(M,N)}$\textit{-trivializable} (Figure~\ref{fig:dl}). We also denote $d_{triv}^{(M, N)}(I_Q):=\sup_{\x\in I_Q}{\{d_{triv}^{(M, N)}(\x)\}}.$
\end{definition}

Note that if we set $N=M$, and let $Q=M$ be the only connected component of intersection of $M$ and $N$, then $d_{triv}^{M,M}(I_M)=\di(M,0)$. We say $M$ is \emph{$\delta$-trivializable} if $\delta \geq \di(M,0)$. 

\begin{definition}

Given an interval module $M$ and the diagonal line $\Delta_\x$ for any $\x\in \Real^2$, 
let $\LL(t)=t\cdot(1,1)+0.5\cdot(-c, c)$ as the parameterization of line $\Delta_{\x}$ with $c=x_2-x_1$ being the vertical intercept of $\Delta_{\x}$. 
We define {\emph{1-parameter slice}} of $M$ along $\Delta_\x$, denoted as $M|_{\Delta_\x}$, given by $M|_{\Delta_\x}(t)=M(\LL(t))$.
Define 
$\di^{\Delta}(M, N)=\sup_{\x\in \Real^2}\set{\di(M|_{\Delta_\x}, N|_{\Delta_\x})}.$
\end{definition}

The following theorem has been proved in~\cite{DX18} which also led to a polynomial time algorithm for computing the interleaving distance between two interval modules.
\begin{theorem} \label{di_criteria}

For two interval modules $M$ and $N$, $\di(M, N)\leq \delta$ if and only if both of the following two conditions are satisfied:

(i) $\delta \geq \di^{\Delta}(M, N)$,

\begin{sloppypar}
(ii) $\forall\delta' > \di^{\Delta}(M, N)$, each intersection component of $M$ and $N_{\rightarrow \delta'}$ is either $(M,N_{\rightarrow \delta'})$\textit{-valid} or $\delta_{(M,N_{\rightarrow \delta'})}$\textit{-trivializable}, and each intersection component of $ M_{\rightarrow \delta'}$ and $N$ is either $(N,M_{\rightarrow \delta'})$\textit{-valid} or $\delta_{(N,M_{\rightarrow \delta'})}$\textit{-trivializable}.
\end{sloppypar}
\end{theorem}

\begin{remark}
    In this theorem, an intersection component between the unshifted module and the shifted one being valid or trivializable depends on the order of modules. However, note that, we always take the order where the unshifted module is followed by the shifted one. Therefore, when we mention the second condition later in the paper, for brevity, we just say each intersection component (on either/both side of shifting) is valid (or trivializable). 
\end{remark}

From the definition $\dH$, it is easy to see the following two facts.
\begin{fact}\label{claim:dh_union}
$d_H(A, B\cup C)\leq \max\{d_H(A,B), \, d_H(A,C)\}$.
\end{fact}

\begin{fact}\label{claim:dh_subset}
For any $A'\subseteq A$  and $B'\subseteq B$, $d_H(A,B)\leq \max\{d_H(A,B'), \, d_H(B, A')\}.$
\end{fact}

\begin{notation*}
$\square_\delta(\x)=\{\x'\in\Real^2: \|\x-\x'\|_\infty\leq \delta\}.$
\end{notation*}

\begin{corollary}\label{cor:dh_slice}
    For interval modules $M$ and $N$, 
    $d_H(I_M, I_N)\leq \delta\implies \delta \geq d_I^{\Delta}(M, N)$.
\end{corollary}

\begin{proof}
Assume $\dH(I_M, I_N)\leq \delta$, we want to show that $\forall\x\in\Real^2, \di(M|_{\Delta_\x}, N|_{\Delta_\x})\leq \delta$. For the case that $I_{M|_{\Delta_{\x}}}$ and $I_{N|_{\Delta_{\x}}}$ are bounded, $\overline{\s\mathbf{t}}\triangleq I_{M|_{\Delta_{\x}}}$ and $\overline{\uu\vv}\triangleq I_{N|_{\Delta_{\x}}}$ are two line segments on the same diagonal line $\Delta_{\x}$ with $\s,\mathbf{t},\uu,\vv\in \Real^2$ being end points. 
If $\|\s-\mathbf{t}\|\leq \delta$  and $\|\uu-\vv\|\leq \delta$, then $\di(M|_{\Delta_{\x}}, N|_{\Delta_{\x}} )\leq \delta$. If one of them is greater than $\delta$, without loss of generality, say $\|\uu-\vv\| > \delta$, we claim that $\s\in \square_\delta(\uu)$ and $\mathbf{t} \in \square_\delta(\vv)$, which implies $\di((M|_{\Delta_{\x}}, N|_{\Delta_{\x}} ))\leq \delta$.
To show it by contradiction, without loss of generality, assume $\s\notin \square_\delta(\uu)$. 
Let $\s=(s_1, s_2)$ and $\uu=(u_1, u_2)$. Then either $s_1-u_1=s_2-u_2>\delta$ or $u_1-s_1=u_2-s_2 >\delta$. 
For the first case where $s_1-u_1=s_2-u_2>\delta$, because of $\dH(I_M, I_N)\leq \delta$, we have $\exists \y\in I_M, \square_\delta(\y)\ni \uu$, or equivalently, $\exists \y\in I_M, \y\in \square_\delta(\uu)$.
Then we have $\y\leq \uu+\vec{\delta}<\s$ which implies $\uu+\vec{\delta}\in I_M$. However, $\uu+\vec{\delta} < \s\in I_{M|_{\Delta_\x}}$, which contradicts the assumption that $\s$ is the lower end of the segment $I_{M|_{\Delta_\x}}$. 
Similarly, for the second case $u_1-s_1=u_2-s_2 >\delta$, we have $\exists \z\in I_N, \z\in \square_\delta(\s)$ because $\dH(I_M, I_N)\leq \delta$. Then we have $\z\leq \s+\vec{\delta} < \uu$ which implies $\s+\vec{\delta}\in I_N$. However, $\s+\vec{\delta} < \uu\in I_{N|_{\Delta_\x}}$ which contradicts the assumption that $ \uu $ is the lower end of the segment $I_{N|_{\Delta_\x}}$.  

For the special case that $I_{M|_{\Delta_{\x}}}, I_{N|_{\Delta_{\x}}}$ are unbounded, we first introduce the following observation which can be checked easily.
\begin{fact}
    Given an interval $I\in \RB^2$, the following three are equivalent:
    \begin{enumerate}
        \item $(+\infty, +\infty)\in I$ ($(-\infty, -\infty)\in I$), 
        \item $\exists \x\in \Real^2, I\cap \Delta_\x$ has no upper (lower) bound,
        \item $\forall \x\in \Real^2, I\cap \Delta_\x$ has no upper (lower) bound.
    \end{enumerate}
\end{fact}
Now assume that $I_{M|_{\Delta_{\x}}}=I_M\cap \Delta_{\x}$ has no upper bound. By the above observation, we have $I_M\ni (+\infty, +\infty)$. If $\dH(I_M, I_N)< +\infty$, then $I_N\ni (+\infty, +\infty)$, which implies $I_{N|_{\Delta_{\x}}}$ also has no upper bound. Therefore $I_{M|_{\Delta_{\x}}}\mbox{ and } I_{N|_{\Delta_{\x}}}$ are matched perfectly on their upper ends. The rest is to check their lower end to see if it is bounded by $\dH(I_M, I_N)$, which is similar to the previous arguments. 

\end{proof}


We define two preorders on intervals that have certain properties. 
\begin{definition}
For any two intervals $I,J \subseteq \Real^2$:\\
    $I\preceq J$ if $\forall \y\in J, \exists \x\in I, \x\leq \y$;\\     
    $I\succeq J$ if $\forall \y\in J, \exists \z\in J, \z\geq \y$.
\end{definition}

It is easy to check $\preceq$ and $\succeq$ are preorders. We claim they have the following properties.
    
\begin{claim}\label{claim:preorder1}
$L(I)\preceq I$, $U(I)\succeq I$.
\end{claim}
\begin{proof}
This can be proved by definition.     
\end{proof}

\begin{claim}\label{claim:subset_eqi_preorder}
For two intervals $I, J$, 
$$J\subseteq I\iff I\preceq J \And I\succeq J\iff L(I)\preceq L(J) \And U(I)\succeq U(J).$$
\end{claim}

\begin{proof}
For the first ($\iff$):

\begin{equation*}
\begin{aligned}
J\subseteq I &\iff \forall \y\in J, \y\in I\\
 &\iff \forall \y\in J, \exists \x, \z\in I, \x\leq \y\leq \z\\
 &\iff I \preceq J \And I\succeq J.
\end{aligned}
\end{equation*}

For the second ($\iff$):
From $L(I)\subseteq I$ and $L(J)\subseteq J$, Claim~\ref{claim:preorder1} and the first ($\iff$), we have $I \preceq L(I) \And I\succeq L(I)$ and $J \preceq L(J) \And J\succeq L(J)$. Place these preordered pairs on the suitable positions, by transitivity of preorder, we get:

\begin{equation*}
    \begin{aligned}
I\preceq J &\implies& L(I)\preceq I \preceq J\preceq L(J)&\implies L(I)\preceq L(J),\\
  L(I)\preceq L(J) &\implies& I \preceq L(I) \preceq L(J)\preceq J &\implies I\preceq J.
    \end{aligned}
\end{equation*}

Similarly, we can get $I\succeq J \iff U(I)\succeq U(J)$.

\end{proof}
\begin{claim}\label{claim:preorder_shift}
$L(I_{\rightarrow\delta})\preceq L(I^{(+\delta)})$ and 
$U(I_{\rightarrow-\delta})\succeq U(I^{(+\delta)})$. Here $I_{\rightarrow\delta}=\{\x\mid \x+\vec{\delta}\in I\}$.
\end{claim}
\begin{proof}
Note that $L(I^{(+\delta)})\subseteq (L(I))^{(+\delta)}$. By Claim~\ref{claim:subset_eqi_preorder} we have $(L(I))^{(+\delta)}\preceq L(I^{(+\delta)})$. Also $L(I_{\rightarrow \delta})\preceq (L(I))^{(+\delta)}$ since for $\forall \x\in L(I)$, $\y\in \square_{\delta}(\x)$,
we have $L(I_{\rightarrow \delta})\ni \x-\vec{\delta}\leq \y$. So $L(I_{\rightarrow \delta})\preceq L(I^{(+\delta)})$. Symmetrically, we can get $U(I_{\rightarrow -\delta})\succeq U(I^{(+\delta)})$.
    
\end{proof}

\begin{claim}\label{claim:preorder_imply_intersection_bound}
$L(J)\preceq L(I) \implies L(I\cap J)\subseteq L(I)$;
$U(I)\succeq U(J) \implies U(I\cap J)\subseteq U(J)$.
\end{claim}
\begin{proof}
For any $\x\in L(I\cap J)$, $\x\in L(I) \cap L(J)$. Assume $\x\in L(J) \setminus L(I)$. By definition of $L(I)$, we have $\exists \y=(y_1<x_1, y_2<x_2)\in I$. But $L(J)\preceq L(I)\implies \exists \x'\in L(J), \x'\leq \y$. That is $\exists \x'=(x'_1<x_1, x'_2<x_2)\in J$, which is contradictory to the definition of $L(J)$. 

Symmetrically, we can prove $U(I)\succeq U(J) \implies U(I\cap J)\subseteq U(J)$.
\end{proof}

\begin{proposition}\label{prop:di_leq_dh}
For two interval modules $M$ and $N$, we have $\di(M, N)\leq\dH(M, N)$. 
\end{proposition}
\begin{proof}


    Suppose $d_H(I_M, I_N)\leq \delta_0$. That is $\forall \delta > \delta_0, I_N\subseteq (I_M)^{(+\delta)}, I_M\subseteq (I_N)^{(+\delta)}$. We want to show that: 

    (i) Each intersection component of $M$ and $N_{\rightarrow \delta}$ is $(M, N_{\rightarrow\delta})$-valid, and each intersection component of $N$ and $M_{\rightarrow\delta}$ is $(N, M_{\rightarrow\delta})$-valid. 
    (ii) $\delta_0 \geq d_I^{\Delta}(M, N)$

    For (i):
    By Claim~\ref{claim:subset_eqi_preorder} we have
    
    \[
    \left\{
                \begin{array}{ll}
                  L((I_M)^{(+\delta)})\preceq L(I_N)\\ U((I_M)^{(+\delta)}) \succeq U(I_N)
                \end{array}
       \right.
    \quad and \quad
    \left\{
                \begin{array}{ll}
                  L((I_N)^{(+\delta)})\preceq L(I_M)\\ U((I_N)^{(+\delta)}) \succeq U(I_M).
                \end{array}
       \right.
\]

    By Claim~\ref{claim:preorder_shift}, we have
    
    $L(I_{N\rightarrow\delta})\preceq L((I_N)^{(+\delta)})\preceq L(I_M)$ and $U(I_{M\rightarrow-\delta}) \succeq U((I_M)^{(+\delta)}) \succeq U(I_N)$. 
    Then by Claim~\ref{claim:preorder_imply_intersection_bound}, we have: 
    
    $L(I_M\cap I_{N\rightarrow\delta})\subseteq L(I_M) \And U(I_{M\rightarrow-\delta}\cap I_N)\subseteq U(I_N)$. 
    Note that the later one is equivalent to say $U(I_{M}\cap I_{N\rightarrow\delta})\subseteq U(I_{N\rightarrow\delta})$. So we have that each intersection component of $M$ and $N_{\rightarrow\delta}$ is ($M, N_{\rightarrow\delta}$)-valid. Symmetrically, we can show that each intersection component of $N$ and $M_{\rightarrow\delta}$ is ($N, M_{\rightarrow\delta}$)-valid.
    
    (ii) directly follows from Corollary~\ref{cor:dh_slice}.


\end{proof}

\begin{definition}
For any interval module $M$, let $\uu, \vv$ be the top-left corner and bottom-right corner of $I_M$ respectively. Then we define the {\emph{bounding band}} of $M$, $\Delta_M$, to be the band between $\Delta_\uu$ and $\Delta_\vv$. 
\end{definition}

\begin{lemma}\label{lemma:dH_iff}
    For any interval module $M$ and rectangle module $R$, let $C=\Delta_M\cap \Delta_R$. Then we have
    \begin{equation}
        \di(M|_{C}, R|_{C}) \neq \dH(M|_{C}, R|_{C}) \iff \di(M|_{C}, 0)\vee \di(R|_{C},0)<\dH(M|_{C}, R|_{C}).
    \end{equation}
\end{lemma}

\begin{proof}
    
($\Longleftarrow$) is trivial.

For ($\implies$), denote $I=M|_{C}, J=R|_{C}$ to simplify the notation.
We want to prove the contrapositive statement. 
Assume $\di(I, 0)\vee \di(J,0)\geq\dH(I, J)$. If we can show that $\di(I, J) \geq \dH(I, J)$, 
then combined with Proposition~\ref{prop:di_leq_dh}, one can get $\di(I, J) = \dH(I, J)$.

Let $\dH(I, J)=\epsilon$. We want to show, $\forall\delta<\epsilon$, $I\mbox{ and } J$ are not $\delta$-interleaved. By assumption, we have $\di(I, J) \geq \epsilon > \delta$. 
Then, at least one of $\di(I, 0)$ and $\di(J, 0)$ is greater than or equal to $\epsilon$. 
Without loss of generality, say $\di(I,0)\geq\epsilon$. Then, 
\begin{equation}\label{eq:nonzero_uu}
\exists\uu\in\Real^2, I(\uu\rightarrow\uu+2\vec{\delta})\neq 0.
\end{equation}

Now we claim that, if $I, J$ are $\delta$-interleaved, then the two intersection components of $I, J_{\rightarrow \delta}$ and $J, I_{\rightarrow \delta}$ are non-empty and valid. 
Note that there is at most one intersection component for either $I, J_{\rightarrow \delta}$ or $J, I_{\rightarrow \delta}$.

The main observation is that, if either of these two intersection components, say the intersection component of $I, J_{\rightarrow \delta}$, is either empty or not valid, then $\phi:I\rightarrow J_{\rightarrow\delta}$ is the zero morphism. That means, by composing with the other interleaving morphism, $\psi(\uu+\vec{\delta})\circ\phi(\uu)=0=I(\uu\rightarrow\uu+2\vec{\delta})$, which is contradiction to Equation~\ref{eq:nonzero_uu}.

Now, by definition of the property of nonempty valid intersection component, we have 
\begin{align*}
    L(J_{\rightarrow \delta})\preceq L(I),\\
    U(J_{\rightarrow \delta})\preceq U(I),\\
    L(I_{\rightarrow \delta})\preceq L(J),\\
    U(I_{\rightarrow \delta})\preceq U(J).
\end{align*}

Observe that by the assumption $C=\Delta_M\cap \Delta_R$,
both $I$ and $J$ touch the upper and lower boundaries of $\Delta_R$. 
Therefore, we have that 
\begin{align}
    U(J_{\rightarrow \delta})\preceq U(I)\implies U(I)\succeq U(J_{\rightarrow \delta}),\\
    U(I_{\rightarrow \delta})\preceq U(J)\implies U(J)\succeq U(I_{\rightarrow \delta}).
\end{align}

Now check that 
\begin{equation*}
L(J_{\rightarrow \delta})\preceq L(I), U(J)\succeq U(I_{\rightarrow \delta})\implies I\subseteq J^{(+\delta)}
\end{equation*}
and
\begin{equation*}
L(I_{\rightarrow \delta})\preceq L(J), U(J)\succeq U(I)\succeq U(J_{\rightarrow \delta})\implies J\subseteq I^{(+\delta)}.
\end{equation*}
That means, $\delta\geq\dH(I,J)=\epsilon$. However, $\delta<\epsilon$ by assumption reaching a
contradiction. So $I, J$ are not $\delta$-interleaved, which establishes the required contradiction. 
\end{proof}


\begin{propositionof}{\ref{thm:form_di}}
    For $I_R$ contained in the bounding rectangle of $I_M$,
    \begin{equation*}
        \di(M, R)=\max\{\di({M|_{\bar{\Delta}_R}}, 0),  (\di(M|_{\Delta_R}, 0)\vee \di(R|_{\Delta_R}, 0)) \wedge \dH(M|_{\Delta_R}, R|_{\Delta_R})\}.    
    \end{equation*}
\end{propositionof}

\begin{proof}
Recall that $\bar{\Delta}_R=\Real^2\setminus \Delta_R$. We have
\begin{equation}\label{eq:di_1}
\di(M,R)=\max\{\di(M|_{\Delta_R}, R|_{\Delta_R}), \di(M|_{\bar{\Delta}_R}, 0)) \}.    
\end{equation}

Based on Lemma~\ref{lemma:dH_iff}, we have
\begin{equation*}
\di(M|_{\Delta_R}, R|_{ \Delta_R})=\min\{ \di(M|_{\Delta_R}, 0)\vee \di(R|_{\Delta_R}, 0),\quad \dH(M|_{\Delta_R}, R|_{\Delta_R})\}.
\end{equation*}
Combined with Equation~\ref{eq:di_1}, we get the result 
\begin{equation*}
    \di(M, R)=\max\{\di({M|_{\bar{\Delta}_R}}, 0),\,  (\di(M|_{\Delta_R}, 0)\vee \di(R|_{\Delta_R}, 0)) \wedge \dH(M|_{\Delta_R}, R|_{\Delta_R})\}.
\end{equation*} 
\end{proof}

\begin{figure}[ht]
    \centerline{\includegraphics[width=\textwidth]{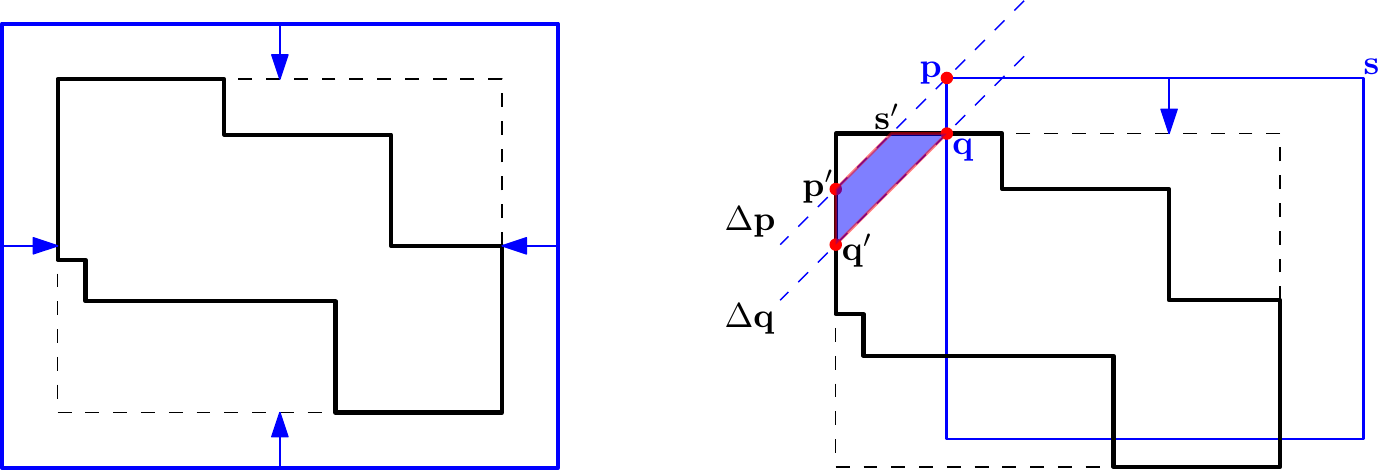}}
    \caption{The left figure illustrates how to push an edge of $I_M$ beyond the bounding rectangle of $I_M$ to the corresponding edge. It is not necessary that each edge of $I_R$ is outside of the bounding rectangle. The right figure is an illustration for the case that $\di(M|_{\bar{C}'})$ becomes larger. Here the blue region indicates the increased region ${\bar{C'}\setminus\bar{C}}$. } 
    \label{fig:opt_rectangle_in_bounding_rectangle}
\end{figure}

\begin{proposition}\label{prop:opt_rectangle_in_bounding_rectangle}
Given a module $M$ with underlying interval $I_M$ and bounding rectangle $B$, and a rectangular module $R$ with $I_R\nsubseteq B$, there exists a rectangle module $R'$ with $I_R'\subseteq B$ such that $\di(M, R')\leq \di(M, R)$. 
\end{proposition}
\begin{proof}
The main idea is that, we can push those edges of $I_R$ that are outside of $B$ to edges of $I_M$. See Figure~\ref{fig:opt_rectangle_in_bounding_rectangle} on the left as an illustration.
Let $C= \Delta_R\cap \Delta_M$ and $\bar{C}$ be its complement. 
We have $\di(M, R)=\max\{\di(M|_{\bar{C}}, 0), \di(R|_{\bar{C}}, 0), \di(M|_C, R|_C)\}$. 
From Lemma~\ref{lemma:dH_iff}, we have $\di(M|_C, R|_C)=\min\{(\di(M|_C, 0) \wedge \di(R|_C, 0)), \, d_H(M|_C, R|_C)\}$. Therefore, we have
\begin{equation}
    \di(M, R)=\min\{\di(M, 0)\vee\di(R, 0), \quad \di(M|_{\bar{C}}, 0)\vee\di(R|_{\bar{C}}, 0)\vee\dH(M|_{C}, R|_{C})\}.
\end{equation}
If $\di(M, R)=\di(M, 0)\vee\di(R, 0)$, we can just choose $R'=0$. 

Assume $\di(M, R)=\di(M|_{\bar{C}}, 0)\vee\di(R|_{\bar{C}}, 0)\vee\dH(M|_{C}, R|_{C})$. After pushing an edge towards the corresponding edge of $B$, 
we get a rectangle module $R'$ with smaller $I_{R'}\subseteq I_R$. For example, if the top edge $[\p, \s]$ of $I_R$ is an edge above the top edge of $B$, then we push it onto the top edge of the bounding rectangle to get edge $[\p', \s']$. 
Let $C'=\Delta_{R'}\cap \Delta_M$. 
Let $\Delta_{\p}, \Delta_{\q}$ be the diagonal lines passing through $\p, \q$ respectively. Let $\p'=L(I_M)\cap \Delta_{\p}$, $\q'=L(I_M)\cap \Delta_{\q}$ and $\s'=U(I_M)\cap \Delta_{\p}$.  See Figure~\ref{fig:opt_rectangle_in_bounding_rectangle} on the right for an illustration.
Then we have $\di(M|_{\bar{C}'}, 0)\geq \di(M|_{\bar{C}}, 0)$, $\di(R|_{\bar{C}'}, 0)\leq \di(R|_{\bar{C}}, 0)$, $\dH(M|_{C'}, R|_{C'})\leq \dH(M|_{C}, R|_{C})$. 
So, the only term possibly becomes larger is $\di(M|_{\bar{C}'}, 0)$. However, that means
$\di(M|_{\bar{C}'}, 0)=\di(M|_{\bar{C'}\setminus\bar{C}})$ is obtained over the region $\bar{C'}\setminus\bar{C}$. Note that the largest possible increase in this region is realized by
a trapezoid on vertices $\p', \q', \q, \s'$. That means, $\p', \q'$ are on the same vertical line and $\s', \q'$ are on the same horizontal line. See Figure~\ref{fig:opt_rectangle_in_bounding_rectangle} on the right for an illustration. 
However, one can check that $\di(M|_{\bar{C'}\setminus\bar{C}})\leq \dl(\p', I_R)$, which implies $\di(M|_{\bar{C'}\setminus\bar{C}})\leq d_H(M|_C, R|_C)$. Therefore, we have $\di(M|_{\bar{C}'}, 0)\leq d_H(M|_C, R|_C)$, which implies $\di(M|_{\bar{C}'}, 0)\leq \di(M|_{\bar{C}}, 0)$. 
\end{proof}

\begin{theoremof}{\ref{prop:properties_box}}
The module $\boxed{M}$ obtained by {\sf Construction 1} satisfies the following properties:
\begin{enumerate}
    \item\label{property1} $\di(M, \boxed{M})\leq \epsilon_M$. 
    \item If $\di(M, 0)\geq \epsilon_M$, then $\di(M, \boxed{M})=\epsilon_M$.
    \item If both $I_M^{\top}$ and $I_M^{\bot}$ are large enough to contain an $2\epsilon_M\times2\epsilon_M$ square, 
    then $\boxed{M}$ is an optimal rectangle module approximating $M$ with respect to $\di$.
\end{enumerate}
\end{theoremof}
\begin{proof}
For $R=\boxed{M}$, we can check that $\dH(M, R)=\epsilon_M$.
\begin{align*}
\dH(M, R)=&\inf_{\delta\geq 0}\{I_M\subseteq {I_R}^{(+\delta)} \And I_R\subseteq I_M^{(+\delta)}  \}\\
=& \inf_{\delta\geq 0} I_M\subseteq {I_R}^{(+\delta)} \vee \inf_{\delta\geq 0} I_R\subseteq I_M^{(+\delta)} \\
=& (\|\rr- \rr'\|_\infty \vee \|\s- \s'\|_\infty )\vee ( \|\rr- \rr''\|_\infty \vee \|\s- \s''\|_\infty)\\
=&\epsilon_M.
\end{align*}
Then immediately Property~\ref{property1} follows from Proposition~\ref{prop:di_leq_dh}, and Property~2 follows from Proposition~\ref{thm:form_di}.



Property 3: By Property 2, we have $\epsilon_M=\di(M, \boxed{M})$.
Let $R^*=[\rr^*, \s^*]$ be an optimal rectangle module with top-left corner and bottom-right corner being $\p^*$ and $\q^*$ respectively.  By Proposition~\ref{prop:opt_rectangle_in_bounding_rectangle}, $R^*$ is contained in the bounding rectangle $[\rr', \s']$ of $I_M$. Let $\p'$ be the top-left corner of $I_M$ and $\q'$ be the bottom-right corner of $I_M$.
By Proposition~\ref{thm:form_di}, we have $\di({M|_{\bar{\Delta}_{R^*}}}, 0)\leq \di(M, R^*)\leq \epsilon_M$ where  $\bar{\Delta}_{R^*}=\Real^2\setminus\Delta_{R^*}$.
By the assumption, 
one can see that $\p^*$ must be in the $\epsilon_M$-box of $\p'$, that is $\p^*\in \square_{(+\epsilon_M)}(\p')$, and similarly, $\q^*\in \square_{(+\epsilon_M)}(\q')$, which ensure that $\di({M|_{\bar{\Delta}_{R^*}}}, 0)\leq \epsilon_M$. From these facts, we have that $\di({M|_{\Delta_{R^*}}}, 0)\geq\epsilon_M$ and also $\rr^*\in \square_{(+\epsilon_M)}(\rr')$ and $\s^*\in \square_{(+\epsilon_M)}(\s')$ (see Figure~\ref{fig:construction1} for an illustration). Now observe that $\dl(\rr^*, L(I_M))\geq \|\rr-\rr''\|_\infty$ and $\dl(\s^*, U(I_M))\geq \|\s-\s''\|_\infty$. 
This means $\dH(M|_{\Delta_{R^*}}, R^*|_{\Delta_{R^*}})\geq \|\rr-\rr''\|_\infty \vee \|\s-\s''\|_\infty = \epsilon_M$. By 
Proposition~\ref{thm:form_di}, we have $\di(M, R^*)
'\geq\epsilon_M$. With $R^*$ being optimal, $\di(M, R^*)\leq \epsilon_M$, which implies $\di(M, R^*)=\di(M, \boxed{M})$. 
\end{proof}

\section{Missing proofs in Section~\ref{sec:general_approx}}\label{sec:missing_proofs2}
Proof details of $\tdi^{\emptyset}=\di$ in Proposition~\ref{prop:tdi=di}.
\begin{proof}
Note that $\di\leq \tdi^{\emptyset}$ by definition. We want to show $\tdi^{\emptyset}\leq \di$.

Assume $M, N$ are $\delta-$interleaved with interleaving morphisms $\phi:M\rightarrow N_{\rightarrow \vec{\delta}}$ and $\psi:N\rightarrow M_{\rightarrow \vec{\delta}}$.  Recall that means $\forall\x\leq \y\in \Real^2$,
    \begin{align}
        &\psi(\x+\vec{\delta})\circ\phi(\x)=M(\x\rightarrow \x+2\vec{\delta}) \label{eq:delta_interleaved_property1}\\
        &\phi(\x+\vec{\delta})\circ\psi(\x)=N(\x\rightarrow \x+2\vec{\delta})\label{eq:delta_interleaved_property2}\\  
        &N(\x+\vec{\delta}\rightarrow \y+\vec{\delta})\circ\phi(\x)= \phi(\y)\circ M(\x\rightarrow\y) \label{eq:delta_interleaved_property3}\\ 
        &M(\x+\vec{\delta}\rightarrow \y+\vec{\delta})\circ\psi(\x)= \psi(\y)\circ N(\x\rightarrow\y) \label{eq:delta_interleaved_property4}
    \end{align}
%
We claim $M^{\aaa}, N^{\aaa}$ are $\delta-$interleaved by $\delta$-interleaving $\phi^{\aaa}, \psi^{\aaa}$ between $M^{\aaa}, N^{\aaa}$ constructed as follows: 
    
    \begin{flalign*}
        \phi^{\aaa}(\x):=&N(\aaa\odot\x+\vec{\delta}\rightarrow \aaa\odot\x+\aaa\odot\vec{\delta})\circ\phi(\aaa\odot\x) \\
        \psi^{\aaa}(\x):=&M(\aaa\odot\x+\vec{\delta}\rightarrow \aaa\odot\x+\aaa\odot\vec{\delta})\circ\psi(\aaa\odot\x)
    \end{flalign*}
    Note that $\aaa$ satisfies $a_1\wedge a_2=1$. So, $\aaa\odot\x+\vec{\delta}\leq \aaa\odot\x+\aaa\odot\vec{\delta}$. Therefore, all the arrows $\rightarrow$ above are well-defined.
    The construction for $\phi^{\aaa}, \psi^{\aaa}$ are symmetric. We give the proof of the interleaving property as following two equations for one side. The other side is symmetric. 
    \begin{equation}\label{eq:interleaving_property1}
        \psi^\aaa(\x+\vec{\delta})\circ\phi^\aaa(\x)=M^\aaa( \x\rightarrow \x+2\vec{\delta})
    \end{equation}
    \begin{equation}\label{eq:interleaving_property2}
        N^{\aaa}(\x+\vec{\delta}\rightarrow\y+\vec{\delta})\circ\phi^{\aaa}(\x)=\phi^{\aaa}(\y)\circ M^{\aaa}(\x\rightarrow\y)
    \end{equation}
    


    
    First check that $\phi^\aaa$ is a collection of maps from $M^\aaa$ to $N^\aaa_{\rightarrow\vec{\delta}}$ from the following diagram (note that $N^\aaa_{\rightarrow \vec{\delta}}(\x)=N^\aaa(\x+\vec{\delta})=N(\aaa\odot\x + \aaa\odot \vec{\delta})$).

\begin{tikzcd}
    M(\mathbf{a}\odot \mathbf{x})=M^\aaa(\x) \arrow[rd, "\phi(\aaa\odot \x)"'] \arrow[rrrrd, "\phi^\aaa(\x)"] &                                                                                      &  &  &                                                                \\
                                                                                                  & N(\aaa\odot \x+{\vec{\delta}}) \arrow[rrr, "N(\aaa\odot \x+{\vec{\delta}} \rightarrow \aaa\odot \x+ \aaa\odot{\vec{\delta}})"'] &  &  & N(\aaa\odot \x+\aaa\odot{\vec{\delta}} )=N^\aaa_{\rightarrow {\vec{\delta}}}(\x)
\end{tikzcd}

Similarly,  $\psi^\aaa$ is a collection of maps from $N^\aaa$ to $M^\aaa_{\rightarrow\vec{\delta}}$.


    To show Equation~\ref{eq:interleaving_property1}, check that
    \begin{align*}
        \psi^\aaa(\x+\vec{\delta})\circ\phi^\aaa(\x)=&M(\aaa\odot\x+\aaa\odot\vec{\delta} +\vec{\delta}\rightarrow \aaa\odot\x+\aaa\odot\vec{\delta}+\aaa\odot\vec{\delta})\\
        &\circ\psi(\aaa\odot\x+\aaa\odot\vec{\delta})\circ N(\aaa\odot\x+\vec{\delta}\rightarrow \aaa\odot\x+\aaa\odot\vec{\delta})\circ\phi(\aaa\odot\x)\\
        =&M(\aaa\odot\x+\aaa\odot\vec{\delta} +\vec{\delta}\rightarrow \aaa\odot\x+\aaa\odot\vec{\delta}+\aaa\odot\vec{\delta})\\
        \circ&N(\aaa\odot\x+\vec{\delta}\rightarrow \aaa\odot\x+\aaa\odot\vec{\delta})\circ\psi(\aaa\odot\x+\aaa\odot\vec{\delta})
        \circ\phi(\aaa\odot\x).
    \end{align*}
    Note that by commutative property of $\psi$, the middle two terms can be changed as follows:
    \begin{equation*}
        \psi(\aaa\odot\x+\aaa\odot\vec{\delta})\circ N(\aaa\odot\x+\vec{\delta}\rightarrow \aaa\odot\x+\aaa\odot\vec{\delta})=
        M(\aaa\odot\x+\vec{\delta}+\vec{\delta}\rightarrow \aaa\odot\x+\aaa\odot\vec{\delta}+\vec{\delta})\circ\psi(\aaa\odot\x+\vec{\delta}).
    \end{equation*}
    Based on this equation, we have
    \begin{align*}
        \psi^\aaa(\x+\vec{\delta})\circ\phi^\aaa(\x)=&M(\aaa\odot\x+\aaa\odot\vec{\delta} +\vec{\delta}\rightarrow \aaa\odot\x+\aaa\odot\vec{\delta}+\aaa\odot\vec{\delta})\\
        \circ&M(\aaa\odot\x+\vec{\delta}+\vec{\delta}\rightarrow \aaa\odot\x+\aaa\odot\vec{\delta}+\vec{\delta})\circ\psi(\aaa\odot\x+\vec{\delta})
        \circ\phi(\aaa\odot\x)\\
        =&M(\aaa\odot\x+2\vec{\delta}\rightarrow \aaa\odot(\x+2\vec{\delta}))\circ
        M(\aaa\odot\x\rightarrow\aaa\odot\x+2\vec{\delta})\\
        =&M(\aaa\odot\x\rightarrow\aaa\odot(\x+2\vec{\delta}))\\
        =&M^\aaa(\x\rightarrow\x+2\vec{\delta}).
    \end{align*}

To verify Equation~\ref{eq:interleaving_property1}, check that:
\begin{align*}
&N^{\aaa}(\x+\vec{\delta}\rightarrow\y+\vec{\delta})\circ\phi^{\aaa}(\x)\\
=&N(\aaa\odot\x+\aaa\odot\vec{\delta}\rightarrow \aaa\odot\y+\aaa\odot\vec{\delta})\circ N(\aaa\odot\x+\vec{\delta}\rightarrow \aaa\odot\x+\aaa\odot\vec{\delta})\circ\phi(\aaa\odot\x)\\    
=&N(\aaa\odot\x+\vec{\delta}\rightarrow \aaa\odot\y+\aaa\odot\vec{\delta})\circ\phi(\aaa\odot\x)\\
=&N(\aaa\odot\y+\vec{\delta}\rightarrow \aaa\odot\y+\aaa\odot\vec{\delta})\circ N(\aaa\odot\x+\vec{\delta}\rightarrow \aaa\odot\y+\vec{\delta})\circ\phi(\aaa\odot\x)\\
=&N(\aaa\odot\y+\vec{\delta}\rightarrow \aaa\odot\y+\aaa\odot\vec{\delta})\circ 
\phi(\aaa\odot\y)\circ M(\aaa\odot\x\rightarrow\aaa\odot\y)\\
=&\phi^{\aaa}(\y)\circ M^{\aaa}(\x\rightarrow\y).
\end{align*}
\end{proof}

\begin{propositionof}{\ref{prop:linearly_orderd_interval_decomp}}
If $P$ is a presentation of $M$ such that neither $\gr(\col P)$ nor $\gr(\row P)$ has incomparable pairs, then $M$ is interval decomposable. 
\end{propositionof}

\begin{proof}
This is a special case which can be handled by the algorithm in~\cite{DeyCheng19} to compute a diagonalization of the presentation matrix $P$. 

The basic idea is as follows. Order the columns and rows in ascending order. One can first run the traditional persistence algorithm from left to right on columns. 
The result is an equivalent matrix $P'\sim P$
such that the lowest indices of nonzero entries (if exist) of all columns are distinct. After that, apply the traditional persistence algorithm to $P'$ from bottom to top on rows. The resulting matrix $Q\sim P$ is a diagonal matrix and the persistence module built on $Q$ is isomorphic to $P$. However, like the fact that a linear map represented by a diagonal matrix can be decomposed as a direct sum of linear maps on one dimensional spaces, 
a persistence module built on a finite presentation with diagonal matrix can be decomposed as a direct sum of indecomposable modules each with a single generator. Such indecomposable
modules have to be interval modules. 
\end{proof}

\begin{propositionof}{\ref{lm:presentation_on_band}}
Given a presentation of $M$ with finite presentation $P$, for any convex set $C\subseteq\Real^2$,
the finite presentation of $M|_C$ is the same matrix as $P$'s with grades $\mathbf{u}$ of rows or columns replaced by $\push_C(\mathbf{u})\triangleq\min\{\ww\in C\mid\ww\geq \uu\}$.
\end{propositionof}
\begin{proof}
Let $Q$ be the finite presentation of $M|_C$ as described in the statement.
Observe that for each $\uu\in C$, $\push_C(\gr(\row_i))\leq \mathbf{u}\}\iff \gr(\row_i)\leq\uu$. Therefore, by the construction of $M^P \mbox{ and }M^{Q}$, $\forall \uu\in C, \gen^{Q}(\mathbf{u})=\gen^{P}(\uu)$. Similarly, 
$\forall \uu\in C, \rel^{Q}(\mathbf{u})=\rel^{P}(\uu)$.
Therefore, $M^Q(\mathbf{u})=\spanning(\gen^{Q}(\mathbf{u}))/\spanning(\rel^{Q}(\mathbf{u}))=\spanning(\gen^{P}(\mathbf{u}))/\spanning(\rel^{P}(\mathbf{u}))=M^P(\mathbf{u})$ for each $\mathbf{u}\in C$ and the induced linear maps are the same. That means, $Q$ defined in the statement is indeed a finite presentation of $M|_C$.
\end{proof}



\begin{figure}[htbp]
\centerline{\includegraphics[width=0.7\textwidth]{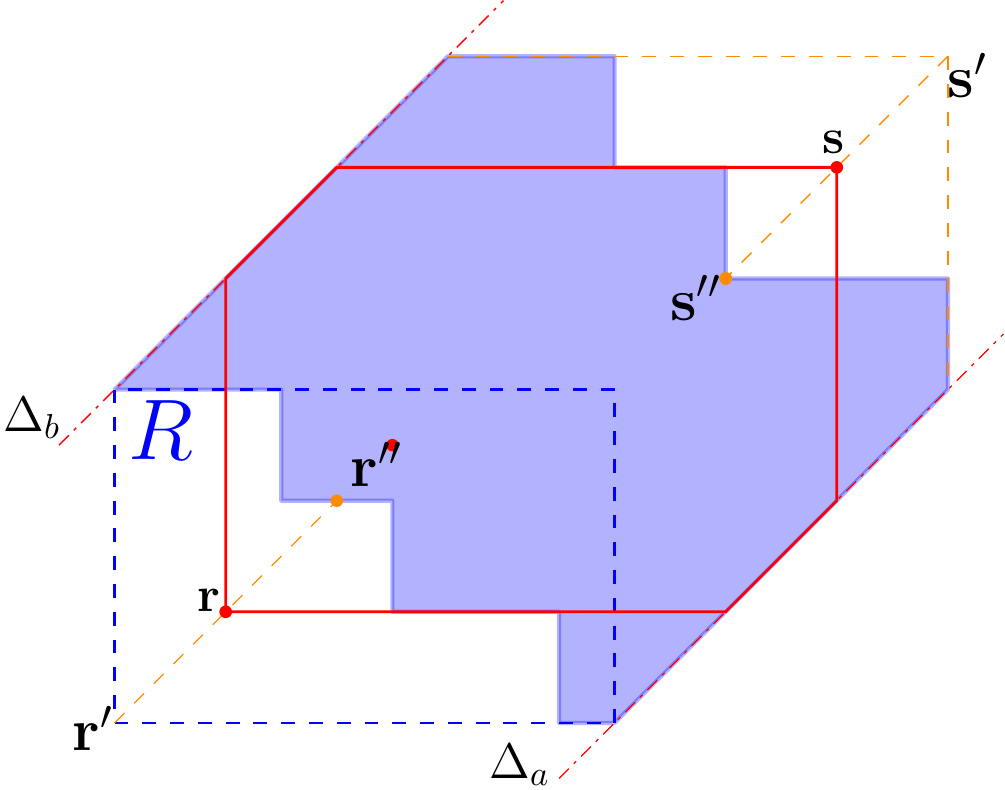}}
\caption{Blue region is $I_M$. The blue dashed rectangle $R$ is the minimal bounding rectangle over $L(I_M)$.  $\dl(\rr', U(R))\leq \frac{1}{4}\epsilon_C$ since this rectangle $R$ is anchored by the diagonal lines $\Delta_a$ and $\Delta_b$.}
\label{fig:epsilon_bound_over_boxed_M}
\end{figure}

\begin{proposition}\label{prop:subdivide_bound}
Let $C=\Delta_{[a,b]}\subset \Real^2$ be a band. Let $\epsilon_C=b-a$. For any interval $C$-modules $M$, $\epsilon_M^*\leq \frac{1}{4}\epsilon_C$.
\end{proposition}
\begin{proof}
Recall that $\epsilon_M^*=\di(M, \boxed{M}^*)\leq \di(M, \boxed{M})=\epsilon_M$.
By definition~\ref{def:construction1}, $\epsilon_M=\frac{1}{2}\min\{\dl(\rr', L(I_M)), \dl(\s', U(I_M))\}$ where $[\rr', \s']$ is the smallest bounding rectangle of $I_M$. 
Consider the minimal bounding rectangle $R$ of $L(I_M)$. 
Observe that the lower left corner of this rectangle coincides with $\rr'$ and $\dl(\rr', L(I_M))\leq \dl(\rr', U(R))$. See Figure~\ref{fig:epsilon_bound_over_boxed_M} . Also it is not hard to see that $\dl(\rr', U(R))\leq \frac{1}{4}\epsilon_C$. 
In a similar way, we can also get $\dl(\s', U(I_M))\leq \frac{1}{4}\epsilon_C$. 
\end{proof}

\end{document}